\documentclass{article}

\usepackage[colorlinks,citecolor=blue,linkcolor=blue,urlcolor=blue,breaklinks]{hyperref}
\usepackage{url}
\usepackage{booktabs}       
\usepackage{amsfonts}       
\usepackage{nicefrac}       
\usepackage{microtype}      
\usepackage{xcolor}         
\usepackage{amsmath,amsfonts,graphicx, amsthm}
\usepackage{appendix}
\usepackage{titletoc}
\usepackage{tikz}
\usepackage[margin = 1in]{geometry}

\newtheorem{assumption}{Assumption}

\newtheorem{lemma}{Lemma}
\newtheorem{theorem}{Theorem}
\newtheorem{remark}{Remark}
\DeclareMathOperator*{\argmin}{arg\,min}
\newcommand\DoToC{%
  \startcontents
  \printcontents{}{1}{\textbf{Contents}\vskip3pt\hrule\vskip5pt}
  \vskip3pt\hrule\vskip5pt
}
\usepackage[round]{natbib}

\title{Spectral Differential Network Analysis\\for High-Dimensional Time Series}

\author{Michael Hellstern$^1$ \and Byol Kim$^2$ \and Zaid Harchaoui$^3$ \and Ali Shojaie$^1$}

\date{\small
	$^1$Department of Biostatistics, University of Washington \\ \texttt{}\\
	$^2$Department of Statistics, Sookmyung Women's University \\ \texttt{} \\
    $^3$Department of Statistics, University of Washington \\ \texttt{} \\ [2ex]
}

\begin{document}

\maketitle

\begin{abstract}
     Spectral networks derived from multivariate time series data arise in many domains, from brain science to Earth science. Often, it is of interest to study how these networks change under different conditions. For instance, to better understand epilepsy, it would be interesting to capture the changes in the brain connectivity network as a patient experiences a seizure, using electroencephalography data. A common approach relies on estimating the networks in each condition and calculating their difference. Such estimates may behave poorly in high dimensions as the networks themselves may not be sparse in structure while their difference may be. We build upon this observation to develop an estimator of the difference in inverse spectral densities across two conditions. Using an $\ell_1$ penalty on the difference, consistency is established by only requiring the difference to be sparse. We illustrate the method on synthetic data experiments and on experiments with electroencephalography data.
\end{abstract}

\section{INTRODUCTION}

Spectral network analysis of multivariate time series plays a key role in fields ranging from oceanography and seismology to neuroscience \citep{ laurindo2019cross, james2017improved, bloch2022network}. We use the inverse spectral density matrix as our choice of network as the $i,j$ entry more directly encodes the dependence between nodes $i$ and $j$ compared to other networks in the spectral domain such as coherence \citet{dahlhaus2000graphical}. Studying individual networks can provide insights into how features interact; however, it is often of interest to study how networks change across conditions or in response to an external intervention \citep{shojaie2021differential}. In neuroscience, for example, many neurodegenerative disorders are associated with abnormal brain connectivity networks \citep{bloch2022network}. Spectral features are regularly used to interpret many types of neuroscientific data, from electroencephalography to magnetoencephalography data~\citep{NEURIPS2023_21718991,NEURIPS2020_de03beff,NEURIPS2018_64f1f27b}.

Coherence, the frequency domain analog to correlation, is a common choice for investigating interactions in multivariate time series analysis. This notion is especially appealing in neuroscience applications, where activities captured at different frequencies better reveal brain oscillations in sensory-cognitive processes. Therefore, despite the availability of nonlinear association measures, such as mutual information \citep{belghazi2018mutual} and transfer entropy \citep{ursino2020transfer}, coherence is commonly used by neuroscientists to define brain functional connectivity networks. 

Similar to correlation, coherence includes the indirect effects of other nodes in the network. Thus, coherence may not be an informative measure of the dependence between nodes. The inverse spectral density is a more direct measure of dependence between nodes as the inverse spectral density between nodes $i$ and $j$ is a rescaling of the coherence after removing the linear effects of all other nodes $\{ k \neq i , j\}$ \citep{dahlhaus2000graphical}. Therefore, the inverse spectral density better resembles the effective connectivity between brain regions \citep{friston2011functional}, providing an initial understanding of how two regions may be causally related.

Modern data collection methods have facilitated the collection of data,  where the dimensionality is much greater than the number of data points. This setting is often referred to as high-dimensional ($p \gg n$). When analyzing high-dimensional time series data, regularization techniques, such as the LASSO, are needed to make the problem computationally and statistically tractable \citep{banerjee2008model}. In the frequency domain, regularization can allow for better numerical stability and better performance \citep{bohm2009shrinkage}.

While several methods exist to directly estimate the difference in inverse covariance matrices (see \citet{tsai2022joint} and references therein), no such methods exist for the difference in inverse spectral densities. Na\"{i}vely, to estimate the difference in inverse spectral densities, one could first estimate the inverse spectral density in each condition and then take their difference. In high-dimensions, consistency of the resulting estimate requires both networks to be sparse \citep{deb2024regularized}. This assumption may be unrealistic due to the presence of hub nodes \citep{wang2021direct} and thereby lead to degraded statistical performance. 

In this paper, we propose Spectral D-trace Difference (SDD), a direct estimator of the difference in inverse spectral densities, which, to our knowledge, is the first method to directly target a differential network in the spectral domain. By leveraging advances in analysis of time series data, our direct estimator only requires sparsity of the difference, which is a more realistic assumption in many settings  \citep{wang2021direct}. For instance, despite identifying many nonzero coherence values in resting state and stimulation states, \citet{bloch2022network} found many coherence changes to be near zero. These small differences are likely due to noise in the estimation procedure, corresponding to no underlying change in coherence. Assuming sparsity of the difference, and only under a geometrically decaying time dependence condition, we show that our direct difference estimator consistently estimates the true difference. Performance of our estimator is studied in a variety of simulation settings, as well as a real data application.

\paragraph{Notation.} We use $\| A \|_1$, $\| A \|_{F}$, $\|A\|_{\infty}$, $\| A \|_{0}$ to denote the $\ell_1$, Frobenius, infinity, and $\ell_0$ norms of a matrix $A$, respectively. 
The element-wise absolute value of a matrix $A$ is denoted as $|A|$. 
The conjugate transpose of a matrix or vector will be denoted as $A^{H}$. 
We denote the inner product between two matrices $A$ and $B$ as $\langle A, B \rangle = \mathrm{Tr}(AB^T)$ and use $ A \otimes B$ to represent their Kronecker product. 

\section{METHOD}

Figure~\ref{fig_methods} provides a visual explanation of our method, SDD, and each step (A)--(E) is elaborated upon in this section.

\begin{figure}[ht!]
  \centering
  \includegraphics[width=0.6\linewidth]{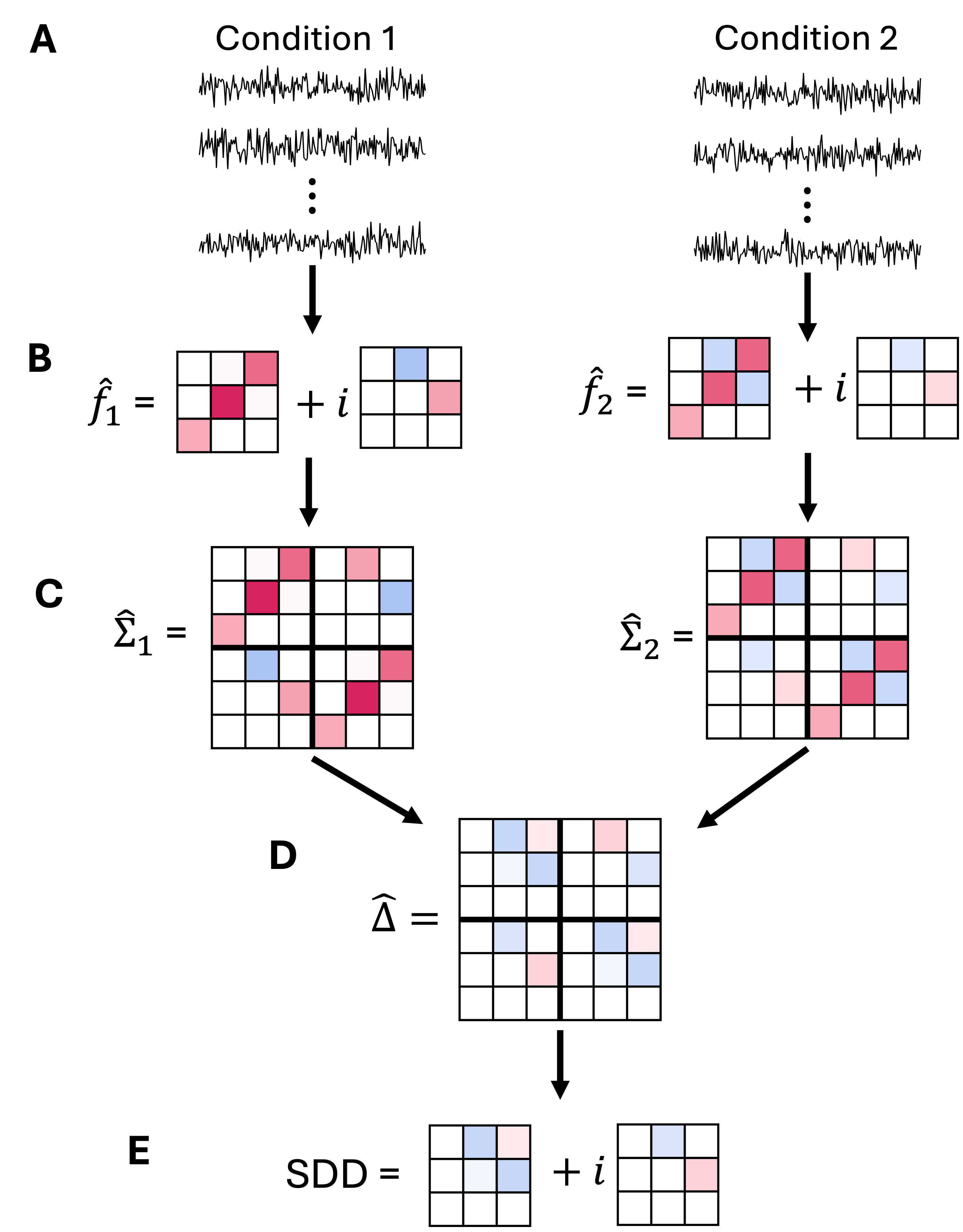}
\caption{SDD Method. \textbf{(A) Inputs:} Time series datasets are observed for two conditions. \textbf{(B) Computing the Spectral Densities:} We then compute the spectral densities in each condition using the smoothed periodogram. \textbf{(C) Expanding to Real Space:} These spectral densities are expanded to the real space. \textbf{(D) Direct Estimation of the Difference in Inverse Spectral Densities:} Using the expanded spectral densities we directly compute the difference estimator. \textbf{(E) Output:} We form the difference in inverse spectral densitites by taking the (1,1) and (2,1) blocks of $\hat{\Delta}$.}
\label{fig_methods}
\end{figure}

\paragraph{(A) Inputs.} Suppose we have mean-zero piecewise stationary data in two conditions, denoted as 1 and 2. In condition 1, there are $n_1$ observations of $p$ covariates, while condition 2 has $n_2$ observations. It is worth noting that the assumption of piecewise stationarity only requires stationarity within each condition and not stationarity across conditions. This may be satisfied in many neuroscience applications where experimental factors affecting the brain state are known from the experimental design. For example, in \citet{bloch2022network}, conditions 1 and 2 might represent pre- and post-stimulation, respectively and stimulation times are known.

\paragraph{(B) Computing the Spectral Densities.} The spectral density in condition $l \in \{1,2\}$ at frequency $\lambda$, $f_l(\lambda)$, is defined as the Fourier transform of the autocovariance matrix. Intuitively, the spectral density represents the portion of the covariance of the signal that can be attributed to the specific frequency $\lambda$ \citep[pp. 331--332]{brockwell1991time}. Mathematically we can write our data as $\{ \mathbf{x}_{l,t} \}$ for $t = 1,\dots, n_l$  where $\mathbf{x}_{l,t}$ is a $p$-dimensional vector in condition $l \in \{1, 2\}$. For each condition $l$, the autocovariance at lag $h \in \mathbb{Z}$ is denoted as $\Gamma_l(h) = \mathbb{E}\left( \mathbf{x}_{l,t} \mathbf{x}^T_{l, t+h} \right)$.  If $\sum_{h=-\infty}^{\infty} \left| \Gamma_l(h) \right| < \infty$ then the spectral density in condition $l$ at frequency $\lambda$ exists and is defined as
\[
f_l(\lambda) = \frac{1}{2\pi} \sum_{h = -\infty}^{\infty} e^{-i\lambda h}\Gamma_l(h), \qquad -\pi \leq \lambda \leq \pi.
\]

To estimate the difference of \emph{inverse} spectral densities, we will require estimates of the spectral densities themselves. We use the smoothed periodogram to estimate the spectral density matrix and refer to it as $\hat{f}_l$. 
To define the smoothed periodogram, we must first define the periodogram. For condition $l$, the periodogram at Fourier frequency $\{\lambda_j = 2 \pi j/n_l, -\lfloor (n_l-1)/2 \rfloor \leq j \leq \lfloor n_l/2 \rfloor \}$ is defined as
\[
P_{l}(\lambda_j) = \frac{\left(\sum_{t=1}^{n_l} \mathbf{x}_{l,t} e^{-i\lambda_j t}\right) \left(\sum_{t=1}^{n_l} \mathbf{x}_{l,t} e^{-i\lambda_j t}\right)^{H}}{2\pi n_l} \, .
\]
Recall $\mathbf{x}^{H}$ is the conjugate transpose of the vector $\mathbf{x}$. Using the Fast Fourier transform (FFT) algorithm, the periodogram can be computed quickly. The periodogram is known to be an inconsistent estimator of the spectral density so a smoothed version is used where the periodograms of $2 M_{n_l}$ nearby frequencies are averaged \citep[pp. 347--350]{brockwell1991time}. $M_{n_l}$ is referred to as the bandwidth or smoothing span. That is, the smoothed periodogram at fourier frequency $\lambda_j$ is generated as
\[
\hat{f}_l(\lambda_j) = \frac{1}{2 M_{n_l} + 1}\sum_{ k = j - M_{n_l}}^{j + M_{n_l}} P_{l}(\lambda_k).
\]

\paragraph{(C) Expanding to the Real Space.} In practice, many different $\lambda$ values are of interest, but for notational simplicity we assume $\lambda$ is fixed throughout and suppress the dependence on $\lambda$. The spectral density, $f_l$, is complex-valued in general which can complicate the estimation.
To transform this problem to the real space,  we note that any complex matrix can be written as the sum of a real matrix and $i$ times another real matrix (for $i^2 = -1$).
Specifically we write the spectral density and its inverse as $f_l = A_l + i B_l$, and $f_{l}^{-1} = \tilde{A}_l + i \tilde{B}_l$,  respectively, where $A_l$, $B_l$, $\tilde{A}_l$, and $\tilde{B}_l$ are real matrices.

With these representations, we can now work in the real space by studying
$
\Sigma_l = \begin{bmatrix} A_l & - B_l \\ B_l & A_l \end{bmatrix} \in \mathbb{R}^{2p \times 2p} 
$ instead of $f_l$. 
By Lemma~A.1 of \citet{fiecas2019spectral}, 
\[
\begin{bmatrix} A_l & - B_l \\ B_l & A_l \end{bmatrix} \begin{bmatrix} \tilde{A}_l & - \tilde{B}_l \\ \tilde{B}_l & \tilde{A}_l \end{bmatrix} = \Sigma_l \Sigma_l^{-1} = I_{2p}, 
\]
where $I_{2p}$ represents the $2p \times 2p$ identity matrix. In this representation,  $\Sigma_l^{-1}$ is the expansion of $f_l^{-1}$ to the real space. Moreover, $f_l$ and $f_l^{-1}$ can easily be recovered by taking the (1,1) and (2,1) blocks of $\Sigma_l$ and $\Sigma_l^{-1}$, respectively.

\paragraph{(D) Direct Estimation of the Difference in Inverse Spectral Densities.} 

Having defined the problem in the real space, we use the $\ell_1$ penalized D-trace loss function \citep{yuan2017differential} to directly estimate $\Delta = \Sigma_1^{-1} - \Sigma_2^{-1}$. 

To justify this estimation strategy, in Appendix~\ref{apdx_dtrace_loss} we show that this loss is convex and the population loss is minimized at the true $\Delta = \Sigma_1^{-1} - \Sigma_2^{-1}$. The inputs to this estimator are $\hat{\Sigma}_1$ and $\hat{\Sigma}_2$ which are formed by expanding $\hat{f}_1$ and $\hat{f}_2$ to the real space respectively. Our estimator is 
\begin{equation}
\label{eqn_dtrace_est}
    \begin{split}
    \hat{\Delta} = & \argmin_{\Delta} \text{\hspace{0.125cm}} \frac{1}{4}\left(\langle \hat{\Sigma}_2 \Delta, \Delta \hat{\Sigma}_1 \rangle + \langle \hat{\Sigma}_1 \Delta, \Delta, \hat{\Sigma}_2 \rangle \right)  - \\
    & \langle \Delta, \hat{\Sigma}_2 - \hat{\Sigma}_1 \rangle  + \tau_{n_1, n_2} \|\Delta\|_1 
    \end{split}
\end{equation}
where $\tau_{n_1, n_2} $ is a penalty parameter depending on sample sizes $n_1, n_2$ and $\|\Delta\|_1$ is the $\ell_1$ penalty which returns the sum of the absolute values of all entries of a matrix. We solve for $\hat{\Delta}$ using the alternating direction method of multipliers (ADMM) algorithm from \citet{yuan2017differential}. Since we are inputting matrices $\hat{\Sigma}_i \in \mathbb{R}^{2p \times 2p}$, the computational complexity of each iteration of this algorithm is $O(8p^3)$ while the memory requirements are $O(4p^2)$.

\paragraph{(E) Output.} Recall that the inverse spectral density, $f_l^{-1}$, represents the coherence between two nodes after removing the linear effects of all other nodes. Our goal is to estimate how this effective connectivity differs between the two conditions. That is, we wish to estimate $f_1^{-1} - f_2^{-1}$.  Using our representation in the real space, the true difference $f_{1}^{-1} - f_2^{-1}$ is recovered by taking the (1,1) and (2,1) blocks of
\[
\Sigma_1^{-1} - \Sigma_2^{-1} = \begin{bmatrix}
\tilde{A}_1 & - \tilde{B}_1 \\
\tilde{B}_1 & \tilde{A}_1
\end{bmatrix} - \begin{bmatrix}
\tilde{A}_2 & - \tilde{B}_2 \\
\tilde{B}_2 & \tilde{A}_2
\end{bmatrix} := \Delta.
\]
Similarly, using our estimate $\hat{\Delta}$ of $\Delta$ we can estimate $f_{1}^{-1} - f_{2}^{-1}$ using the (1,1) and (2,1) blocks of $\hat{\Delta}$. 

In the following section, we study the consistency of our estimator. Intuitively, if $\hat{f}_l$ estimates the spectral density well, we should be able to estimate $\Delta$ well. Our theoretical results show that this is indeed the case and the rate at which the smoothed periodogram $\hat{f}_{l}$ converges to the true periodogram is also the rate at which $\hat{\Delta}$ converges to $\Delta$. Due to the dependence in the data, the smoothed periodogram and $\hat{\Delta}$ are not $\sqrt{n}$ consistent as one might expect in the i.i.d case. Rather, we show that for an optimal choice of smoothing span, $\hat{f}_l$ and $\hat{\Delta}$ converge to $f_l$ and $\Delta$ respectively at a rate of $n^{-1/3}$.

\section{THEORY}\label{sec:theory}
In this section, we study the theoretical properties of our SDD estimator, $\hat{\Delta}$. Specifically, we establish the convergence rate of $\hat{\Delta}$ to the population quantity $\Delta$. From  Equation~\ref{eqn_dtrace_est}, our estimator $\hat{\Delta}$ relies on expanded estimates of the spectral density $\hat{\Sigma}_1, \hat{\Sigma}_2$. Thus, it is natural that the rate of convergence of $\hat{\Delta}$ to $\Delta$ relies on the rates of convergence of $\hat{f}_1$ and $\hat{f}_2$ to $f_1$ and $f_2$.

\begin{theorem} \label{thm_consistency}
     Suppose Assumption~\ref{apdx_assumption_fiecas_dependence} in Appendix~\ref{apdx_functional_dependence} holds and that $\tau_{n_1, n_2} \geq O(\min(n_1,n_2)^{-1/3})$ and $M_{n_l} = O(n_l^{2/3})$. For $n_1$ and $n_2$ large enough and for any $H > 0$, there exists some $C > 0$ such that with probability greater than $1 - Cp^2\min(n_1,n_2)^{-H}$, 
\[
    \| \hat{\Delta} - \Delta \|_{F} \leq O\left(s_{\Delta} \min(n_1,n_2)^{-1/3}\right),
\]
where $s_{\Delta}$ is the number of non-zero elements in the true difference matrix $\Delta$.
\end{theorem}

The full proof of Theorem~\ref{thm_consistency} is given in Appendix~\ref{apdx_proof_consistency}. To establish convergence rates of $\hat{f}_1$ and $\hat{f}_2$ in high-dimensions, we rely in part on Theorem~3.1 of \citet{fiecas2019spectral}, which uses the functional dependency framework of \citet{wu2005nonlinear} and \citet{wu2018asymptotic}. Specifically, our convergence rates hinge on the assumption of a geometrically decaying dependence in the data. An overview of the functional dependence framework is given in Appendix~\ref{apdx_functional_dependence} and a formal statement of the geometric decay assumption is given in Assumption~\ref{apdx_assumption_fiecas_dependence} in Appendix~\ref{apdx_functional_dependence}.  The geometric decay assumption is mild and is satisfied by many processes including ARMA, ARMA-ARCH, and ARMA-GARCH processes, as well as other nonlinear autoregressive processes \citep{shao2007asymptotic, liu2010asymptotics}. 

While our method uses a similar setup to \citet{yuan2017differential}, the analysis of our estimator is fundamentally different in two ways.  First, we establish Theorem~\ref{thm_consistency} only assuming sparsity of the difference and we explicitly avoid the irrepresentability assumption of \citet{yuan2017differential} which has been shown not to hold in practice \citep{zhao2006model}. For example, \citet{wang2021direct} show through simulations that for graphs of moderate or larger dimensions, this condition almost never holds. This is true even for very sparse differences and weak data dependence. Second, in contrast to the proof of \citet{yuan2017differential}, which assumes i.i.d. observations, our proofs also leverage advances in theoretical analysis of high-dimensional time series to handle the temporal dependence among observations over time. It is also worth noting that Theorem~\ref{thm_consistency} also applies to differences in inverse spectral densities across frequencies. 

We now discuss interesting intermediate results from the proof. In Lemma~\ref{lemma_conc_f}, we show that, with high probability and a smoothing span of $M_{n_l} = O(n_l^{2/3})$, the deviations of the smoothed periodogram from the true spectral density in condition $l$ are $O(n_{l}^{-1/3})$. Lemma~\ref{lemma_conc_hessian} shows that with high probability, the estimated Hessian of the D-trace loss converges to the true Hessian at a rate of $O(\min(n_1,n_2)^{1/3})$. That is, $\|\hat{\Gamma} - \Gamma \|_{\infty} = O(\min(n_1,n_2)^{1/3})$ where $\hat{\Gamma} = 0.5 \left( \hat{\Sigma}_1 \otimes \hat{\Sigma}_2 + \hat{\Sigma}_2 \otimes \hat{\Sigma}_1 \right)$ and $\Gamma$ is the same but without the hats. Lastly, Equation~\ref{eqn_dtrace_deriv_norm} establishes that, with high probability, the maximum norm of deviations of the derivative of the D-trace loss, $\left\|0.5\left( \hat{\Sigma}_1 \Delta \hat{\Sigma}_2 + \hat{\Sigma}_2 \Delta \hat{\Sigma}_1 \right) - \left( \hat{\Sigma}_1 - \hat{\Sigma}_2 \right)\right\|_{\infty}$, is $O(\min(n_1,n_2)^{-1/3})$. 

Due to the dependence in the data, the convergence rate in Theorem~\ref{thm_consistency} is slower than the $n^{-1/2}$ rates for i.i.d data. This is because, compared with parametric models that amount to fast decaying  dependence among observations, our dependence assumption is very mild and general. 
In fact, as we choose $M_n$ to minimize the deviations in Theorem~3.1 of \citet{fiecas2019spectral}, our rate is the best we can expect to achieve using this general dependence framework when only assuming a
geometrically decaying time dependence. 

\section{SIMULATIONS}\label{sec:simulations}
In this section, we evaluate the performance of the proposed estimator $\hat{\Delta}$ using simulation studies. The code to reproduce these results as well as those in Section~\ref{sec:eeg} are available at \href{https://github.com/mikehellstern/spectral-differential-networks}{\color{blue}{https://github.com/mikehellstern/spectral-differential-networks}}. All simulations use VAR$(1)$ processes, with $p = 54$, as this allows us to compute the true spectral density using results from \citet{sun2018large}. We control the sparsity of the difference in inverse spectral densities by using a block diagonal structure. In all settings, the transition matrix in condition~1 is the same as that in condition~2 except the last $3 \times 3$ block is multiplied by $-1$. The first setting uses similar coefficients as in \citet{sun2018large}. In the second and third simulation settings, the transition matrix consists of one $51 \times 51$ block and one $3 \times 3$ block. The larger block was generated with 60\% and 95\% sparsity in the second and third settings, respectively. All simulations were performed using $n = 100, 200, 500, 1000, 2000$ observations for both conditions. Smoothed periodograms were computed from data for each condition using a smoothing window of $M_n = \lceil n^{2/3} \rceil$ and then converted to the real space to generate $\hat{\Sigma}_l$. All computations were performed on a MacBook with a 2.3 GHz Quad-Core Intel Core i7 processor and 16 GB of memory. More information on the simulation setup, optimization techniques, and tuning parameter selection can be found in Appendix~\ref{apdx_sim}.

\begin{figure}[ht!]
    \centering
    \includegraphics[width=0.5\linewidth]{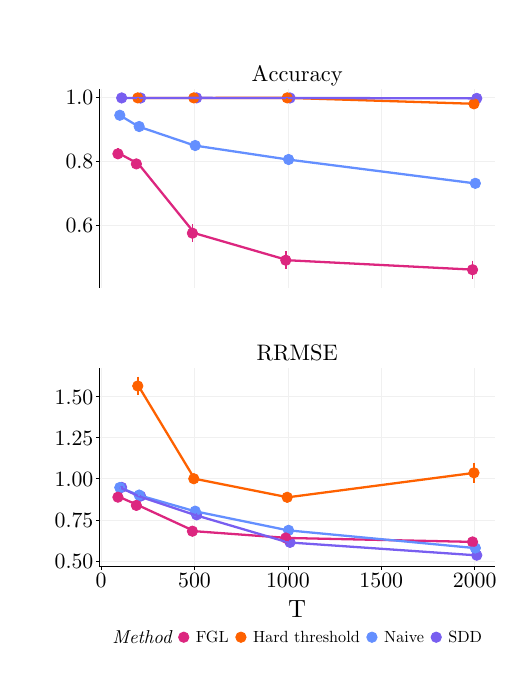}  
    \caption{Simulation 1. Results are reported as mean (dots) and SE (vertical lines) where the mean and SE are taken across all frequencies for a given sample size $n$.}
    \label{fig_sim1_acc_rrmse}
\end{figure}

As previously noted, our method is the first to directly estimate the difference in inverse spectral densities and no direct competitors or state-of-the-art methods exist. However, in this section we will discuss and compare several alternative methods. While not focused on directly estimating the difference, a related method is the joint graphical lasso with a fusion penalty \citep[FGL,][]{danaher2014joint}. This method aims to simultaneously estimate inverses in each condition that are believed to share structural similarities. To encourage similarities in the inverse estimates, FGL uses an $\ell_1$ penalty on their difference. In simulations, FGL was tuned using AIC from \citet{danaher2014joint}. We have also conceptualized two additional methods to compare to SDD. The first is a Na\"ive method, which estimates $\Delta$ by taking the difference after estimating individual inverse spectral densities separately, $\hat{\Delta}_{\mathrm{N}} = \hat{\Sigma}_1^{-1} - \hat{\Sigma}_2^{-1}$. For the Na\"ive method, the graphical LASSO \citep[GLASSO,][]{friedman2008sparse} was used to estimate a sparse inverse spectral density in each condition. Specifically we used the fast implementation from \citet{sustik2012glassofast} available in the \texttt{glassoFast} package in \texttt{R} (GPL ($\geq$ 3.0)). To induce sparsity in individual estimates and (potentially) their difference, GLASSO was tuned for each condition separately using eBIC with $\gamma = 0.5$ \citep{foygel2010extended}. The next method directly thresholds the difference in inverse spectral densities to induce sparse estimates. Specifically, we estimate the inverse spectral density in each condition and take the difference and then use Hard thresholding. The thresholds were tuned using the same eBIC as SDD given in (\ref{eqn_sdd_ebic}). Similar to \citet{deb2024regularized}, all methods are compared based on the following metrics: Accuracy, Precision, Recall, and Relative Root Mean Square Error (RRMSE). The full definitions can be found in Appendix~\ref{apdx_sim}.  

\begin{figure}[ht!]
    \centering
      \includegraphics[width=0.5\linewidth]{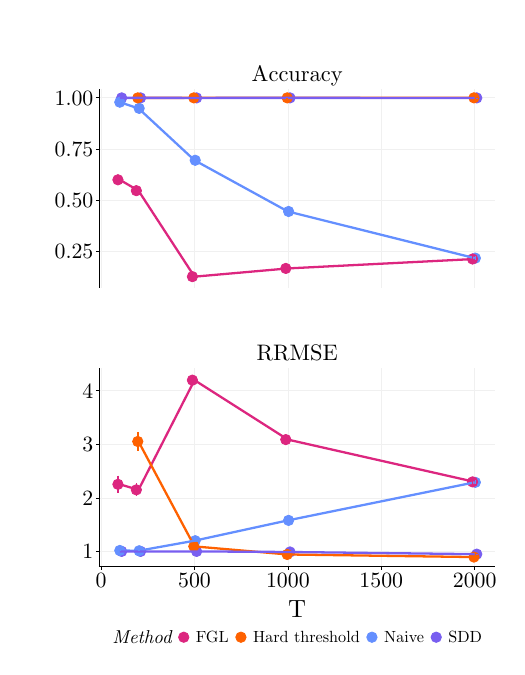} 
    \caption{Simulation 2. Results are reported as mean (dots) and SE (vertical lines) where the mean and SE are taken across all frequencies for a given sample size $n$.}
    \label{fig_sim2_acc_rrmse}
\end{figure}

Simulation results for Accuracy and RRMSE in simulations 1 and 2 are reported in Figures~\ref{fig_sim1_acc_rrmse} and \ref{fig_sim2_acc_rrmse}. Results for Precision and Recall for simulations 1 and 2 can be found in Appendix~\ref{apdx_sim} Figures~\ref{fig_sim1_prec_reca} and \ref{fig_sim2_prec_reca}. Results from simulation 3  are in Appendix~\ref{apdx_sim}, Figures~\ref{fig_sim3_acc_rrmse} and \ref{fig_sim3_prec_reca}.  The full set of results for each of the simulations are in Appendix \ref{apdx_sim} Tables~\ref{tbl_sim1}, \ref{tbl_sim2}, \ref{tbl_sim3}. In general, SDD has better accuracy than both FGL and the Na\"ive method with the difference in accuracy increasing as the sample size increases. 
This is likely due to the fact that, in general, SDD estimates much fewer edges than either method. This also results in fewer false positive edges but more false negatives, resulting in better precision for SDD at the expense of reduced recall; see Table~\ref{tbl_sim1} in Appendix~\ref{apdx_sim}. Across all settings and sample sizes we also see that SDD generally has either lower or similar RRMSE compared to other methods. Overall, compared to FGL and the Na\"ive method, SDD identifies a similar number of true edges and achieves similar or lower RRMSE using a much smaller edge set. Compared to Hard thresholding, SDD achieves much better RRMSE and recall across all settings and sample sizes except setting 2 for large $n$ where Hard thresholding slightly outperforms SDD in terms of RRMSE. To further support our analysis, we have also summarized the difference between SDD and each of the competing methods for each metric in Appendix~\ref{apdx_sim}, Tables~\ref{apdx_tbl_sim1_diff_sdd_naive},~\ref{apdx_tbl_sim1_diff_sdd_hard},~\ref{apdx_tbl_sim1_diff_sdd_soft},~\ref{apdx_tbl_sim2_diff_sdd_naive},~\ref{apdx_tbl_sim2_diff_sdd_hard},~\ref{apdx_tbl_sim2_diff_sdd_soft},~\ref{apdx_tbl_sim3_diff_sdd_naive},~\ref{apdx_tbl_sim3_diff_sdd_hard},~\ref{apdx_tbl_sim3_diff_sdd_soft}.

\section{APPLICATION TO EEG DATA}

\label{sec:eeg}

We apply our direct difference estimator (SDD) and competing estimators to electroencephalograms (EEG) data from \citet{hatlestad2022bids}. Data were recorded with a 64 channel EEG array for 111 healthy subjects at a sampling frequency of 1024 Hz. Four minutes of brain activity was recorded while subjects were resting with their eyes closed. For 42 subjects, a second session was recorded 2--3 months after the initial session. For our analysis, we used the pre-cleaned data provided in OpenNeuro Dataset ds003775 \citep{markiewicz2021openneuro}. Specific cleaning steps can be found in \citet{hatlestad2022bids}. We also downsampled the data to 512 Hz.

To validate our method, we analyzed the subjects with follow-up and without follow-up separately. For those with follow-up, we estimated the difference in networks from the first to the second session (\emph{across session analysis}). For those without, we estimated the difference in networks from 0-60s to 120-180s (\emph{within session analysis}). The 60-120s block was used as a rest. It has been shown that brain networks are temporally dynamic \citep{zalesky2014time, nobukawa2019changes}. Therefore, \textit{a priori}, we expect the differential networks to be sparser in the within session analysis compared to the across session analysis and finding such results would provide validation of our method.

We compare the sparsity of the estimated differential network for SDD, FGL, the Na\"{i}ve, and Hard threshold methods across the commonly-used Theta, Beta, Gamma and High-Gamma bands. Sparsity was defined as the proportion of entries in the estimated difference that were non-zero. Five or six evenly spaced frequencies were considered for each band. They were (in Hz): Theta: (4,5,6,7,8), Beta: (12, 16, 20, 24, 28), Gamma: (30, 40, 50, 60, 70), High-gamma: (80, 95, 110, 125, 140, 150). Results  averaged over subjects and frequencies within each band are presented in Figures~\ref{fig_eeg_within} and ~\ref{fig_eeg_across}.

\begin{figure}[ht!]
    \centering
      \includegraphics[width=0.5\linewidth]{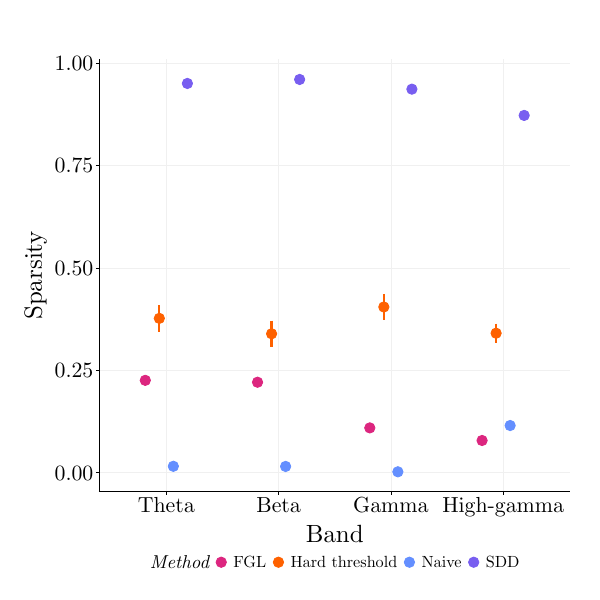} 
    \caption{EEG Within Session Analysis. Vertical lines indicate SE.}
    \label{fig_eeg_within}
\end{figure}

\begin{figure}[ht!]
    \centering
      \includegraphics[width=0.5\linewidth]{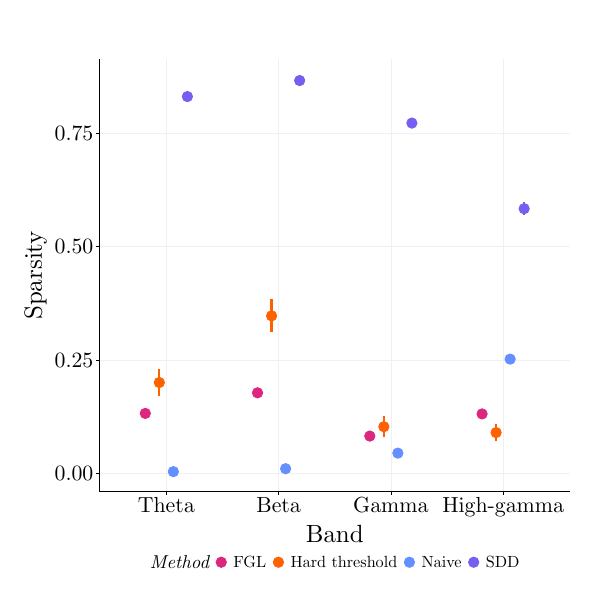} 
    \caption{EEG Across Session Analysis. Vertical lines indicate SE.}
    \label{fig_eeg_across}
\end{figure}

Figures~\ref{fig_eeg_within} and ~\ref{fig_eeg_across} shows that SDD is the only method that estimates a sparser difference in the within session analysis compared to the across session analysis across all frequency bands. On the other hand, competing methods have similar or less sparsity in the within session analysis compared to the across-session analysis. These results clearly indicate that SDD  estimates are more consistent with the experimental setting in this application. 

\section{APPLICATION TO SIMULATED \texorpdfstring{$\boldsymbol{\mu}\mathrm{\textbf{ECoG}}$}{\textbf{uECoG}} DATA}
\label{sec:uecog}

We next apply our SDD estimator to simulated micro-electrocorticography ($\mu\mathrm{ECoG}$) data. The simulation design is based on the optogenetic stimulation $\mu\mathrm{ECoG}$ experiments from \citet{yazdan2016large, yazdan2018targeted, bloch2022network}. Studying how stimulation changes the brain network using such experiments can help design protocols to change the brain network in a targeted way. This can lead to treatments for conditions,  such as schizophrenia or epilepsy, which are known to stem from abnormal brain connectivity networks \citep{bloch2022network}. Our simulations used a similar setup to \citet{yazdan2016large}. Specifically, a total of 32 experimental and 4 control sessions were simulated. Using the same session types as \citet{bloch2022network}, of the 32 experimental sessions, 23 were generated using a simulated 10ms stimulation delay and 9 were generated with a 100ms delay. Each experimental session consisted of five 10-minute stimulation blocks (stimulation state) with a 5-minute recording block (resting state)  before and after. Thus, each experimental session alternated between resting state and stimulation state blocks consisting of 6 total resting state blocks and 5 total stimulation state blocks. We simulated data using a sampling frequency of 500Hz.

Data in all sessions was simulated using a VAR(1) process with i.i.d standard normal errors. It is worth noting that due to the autoregressive nature of the process, autocovariance matrices between variables will not be diagonal in general despite using a diagonal covariance matrix for the errors \citep{lutkepohl2005new}. Since the experiments in \citet{yazdan2016large} included a varying number of good electrodes out of the 96 total electrodes, we randomly selected between $p = 56$ and $p = 94$ good electrodes for each session. Coefficients in a base transition matrix ($A_1 \in \mathbb{R}^{p \times p}$) were generated by randomly sampling from a standard normal distribution. 
To obtain a stable time series, the matrix $A_1$ was modified to have maximum eigenvalue $< 1$. Similar to \citet{yazdan2016large}, two electrodes in each session were randomly chosen to be stimulated. The effect of stimulation was modeled by changing the coefficient in $A_1$ between simulation electrodes which is denoted as $A_{\mathrm{stim}}$. For control sessions, there was no stimulation so $A_1 = A_{\mathrm{stim}}$. The effect of a 10ms and 100ms stimulation delay was simulated by changing the coefficient between stim sites to be $2 \max(|A_1|)$ and $0.5 \max(|A_1|)$, respectively. In this way, we simulate a weaker effect of the 100ms stimulation delay \citep{bloch2022network}.  It is worth noting that, in contrast to the simulations in Section~\ref{sec:simulations}, the transition matrices are dense with no block structure. Thus, it is not immediately clear how the true spectral density changes as the transition matrices change.

During stimulation blocks, data were simulated using the corresponding $A_{\mathrm{stim}}$ matrix. To simulate a lasting, but weakened effect of stimulation on the following resting state, resting state blocks were simulated using a progressively increasing effect of stimulation. Specifically, resting state block 1 uses $A_1$ since no stimulation has occurred. The effect of stimulation for resting state blocks 2 through 6 linearly increases from $\approx0.05$ to $\approx0.6$ times the stimulation effect. For example, for a 10ms delay session, resting state block 2 was simulated using a VAR(1) process with transition matrix equal to $A_1$ where the coefficient between stimulation sites was replaced with $\approx0.05 * 2\max(|A_1|)$.

Using our simulated data, we study how stimulation changes the connectivity between stimulation sites. However, instead of using a na\"{i}ve estimator of the connectivity changes as in \citet{bloch2022network}, we use SDD. Using SDD, we estimate the change in inverse spectral densities from resting state to the following resting state (RS to RS) and from resting state to the following stimulation state (RS to SS). For each session there are 5 RS to RS differences (RS1 to RS2, RS2 to RS3, \dots, RS5 to RS6) and 5 RS to SS difference (RS1 to SS1, \dots, RS5 to SS5). Note that an estimated SDD coefficient of zero indicates no change while non-zero coefficients indicate a change. This could be either a change in magnitude of the connection, a new connection, or a lost connection. The frequencies we study in each band are the same as in Section~\ref{sec:eeg}. Results are averaged over all sessions, frequencies within each band, delay (control, 10ms, or 100ms) and RS to RS or SS to SS.

\begin{figure}[ht!]
  \centering
  \includegraphics[width=0.5\linewidth]{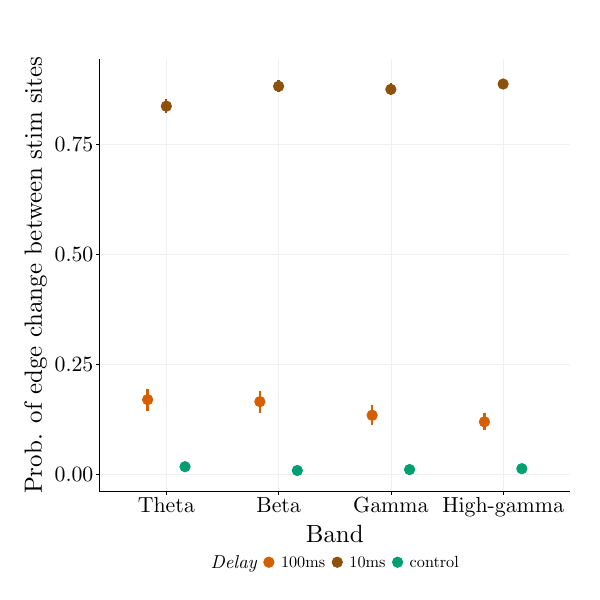}  
\caption{RS to RS analysis. Probability of Change in Edges Between Stimulation Sites by Frequency Band and Delay. Vertical bars indicate standard errors.}
\label{fig_prob_edge_stim_sites_rs_to_rs}
\end{figure}

\begin{figure}[ht!]
  \centering
  \includegraphics[width=0.5\linewidth]{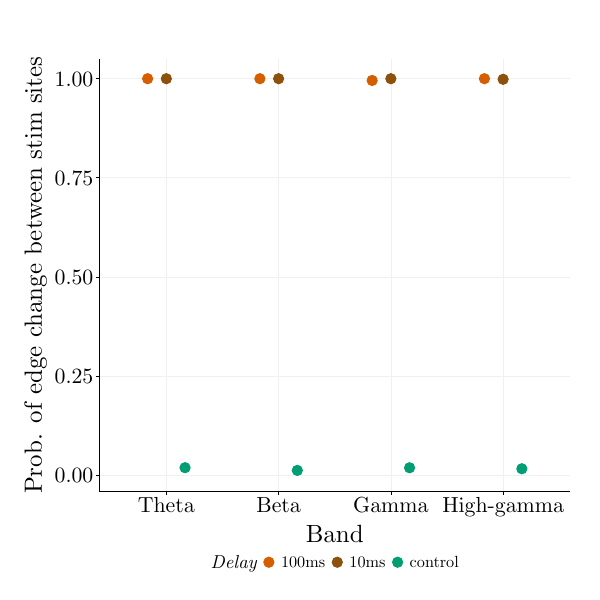}  
\caption{RS to SS analysis. Probability of Change in Edges Between Stimulation Sites by Frequency Band and Delay. Vertical bars indicate standard errors.}
\label{fig_prob_edge_stim_sites_rs_to_ss}
\end{figure}

To study how stimulation changes the edges between stimulation sites, we compute the probability of edge change between the stimulation sites for each band, delay, and RS to RS or RS to SS. The RS to RS and RS to SS results are displayed in Figures~\ref{fig_prob_edge_stim_sites_rs_to_rs} and~\ref{fig_prob_edge_stim_sites_rs_to_ss} respectively. In the RS to SS analysis, stimulation with a 10ms or 100ms delay is more likely to induce a change in connectivity between the stimulation sites across all frequency bands compared to control sessions. After stimulation ends (RS to RS analysis), this connectivity change is attenuated, but still persists. This is in line with our simulation framework and prior work which showed a change in functional connectivity between stimulation sites in the stimulation state which persisted into the following resting state \citep{bloch2022network}.

\section{DISCUSSION}\label{sec:discussion}

The SDD estimator was compared to the joint graphical lasso with a fusion penalty \citep[FGL,][]{danaher2014joint}, the the Na\"ive method of estimating sparse inverse the spectral densities and taking their difference, and directly hard thresholding the difference. The four methods were compared across three different simulation settings, where the true difference in inverse spectral densities was sparse. Comparison metrics included precision, recall, accuracy, and relative root mean square error (RRMSE). When compared to the other methods, SDD performed better in accurately identifying edges and non-edges and the SDD estimates had a lower relative root mean square error. Using EEG data we further validated SDD by showing it estimated sparser differences between networks that were closer in time, in line with expected behavior. %

While consistency of SDD is established in Theorem~\ref{thm_consistency}, this result does not quantify the degree of uncertainty around these estimates. A particularly interesting future direction is developing methods for uncertainty quantification, such as confidence intervals and hypothesis testing. These may be important when the relationship between two brain regions or nodes in the graph is of primary interest or when the sample size is small and noise may be large relative to the signal. 

\section{CONCLUSION}\label{sec:conclusion}
We proposed a direct estimate of the differences in inverse spectral densities between two conditions, termed the Spectral D-trace Difference (SDD) estimator. By using the direct difference estimator, we only need to assume sparsity in the difference of inverse spectral densities. Compared to the usual assumptions of sparsity of each inverse spectral density, the sparsity of the difference is more realistic if, for example, the inverse spectral densities do not change much between conditions. This is indeed what we expect in many biological settings, especially neurodegenerative disorders that are associated with changes in brain connectivity. Convergence rates of our estimator to the true difference were derived using only an assumption on decaying dependence within the time series.

\subsection*{Acknowledgements}

This work was partially supported by U.S. NIH grants R01-MH125429 and R01-GM133848.


\bibliographystyle{plainnat}
\bibliography{references}

\clearpage
\appendix
\section*{Appendix}
\DoToC
\clearpage

\section{ADDITIONAL DETAILS FOR METHODS}

\subsection{D-trace Loss}
\label{apdx_dtrace_loss}
To estimate $\Delta$, we use the $\ell_1$ penalized D-trace loss function from \citet{yuan2017differential}:
\[
L_{D}\left( \Delta, \Sigma_2, \Sigma_1 \right) = \frac{1}{4}\left(\langle \Sigma_2 \Delta, \Delta \Sigma_1 \rangle + \langle \Sigma_1 \Delta, \Delta \Sigma_2 \rangle \right)  - \langle \Delta, \Sigma_2 - \Sigma_1 \rangle \,
\] 
where $\langle A, B \rangle = \mathrm{Tr}(AB^T)$. Similar to \citet{yuan2017differential},  we can take the derivative with respect to $\Delta$ to see that $\Sigma_1^{-1} - \Sigma_2^{-1}$ minimizes the D-trace loss. 

More specifically, let 
\begin{equation} \label{eqn_dtrace_derivative}
\frac{\partial{L_D}}{\partial{\Delta}} = \left( \Sigma_2 \Delta \Sigma_1 + \Sigma_1 \Delta \Sigma_2 \right)/2 - \left( \Sigma_2 - \Sigma_1 \right) \, .
\end{equation}

Evaluating $\frac{\partial{L_D}}{\partial{\Delta}}$ at ${\Delta = \Sigma_1^{-1} - \Sigma_2^{-1}}$ yields  $0$, establishing that $\Delta = \Sigma_1^{-1} - \Sigma_2^{-1}$ minimizes $L_D$. Furthermore, the Hessian of $L_D$ with respect to $\Delta$ is $\frac{\partial^2{L_D}}{\partial{\Delta^2}} = (\Sigma_1 \otimes \Sigma_2 + \Sigma_2 \otimes \Sigma_1)/2$ where $\otimes$ is the kronecker product. Assuming both $\Sigma_1$ and  $\Sigma_2$ are positive definite, $L_D$ is convex in $\Delta$ and $\Delta = \Sigma_1^{-1} - \Sigma_2^{-1}$ is the unique minimizer.

In practice, the population quantites $\Sigma_1$ and $\Sigma_2$ are not available. Instead, we use the estimates $\hat{\Sigma}_1$ and $\hat{\Sigma}_2$. We can additionally incorporate an $\ell_1$ penalty to estimate a sparse $\Delta$ which gives us our estimator in (\ref{eqn_dtrace_est}).

\section{FUNCTIONAL DEPENDENCE FRAMEWORK}
\label{apdx_functional_dependence}

Let $\epsilon^*_{0}$ and $\{\epsilon_t\}_{t \in \mathbb{Z}}$ be i.i.d. vectors in $\mathbb{R}^{b}$ and $\mathcal{F}_t = (\dots, \epsilon_{t-1}, \epsilon_t)$. We further let $\mathbf{x}_{l,t} = \left(x_{l,1t}, \dots x_{l,pt}\right)^T$ be a $p$-dimensional process in condition $l$ where
\[
x_{l, jt} = R_{l,j}(\mathcal{F}_t).
\]

Note that $l$ indexes the condition, in this case 1 or 2, and $j$ indexes the variable, $1, \dots, p$. To measure the dependence of $\mathbf{x}_{l,t}$ on $\epsilon_0$ we can replace $\epsilon_0$ in $\mathcal{F}_t$ with an i.i.d. copy $\epsilon_0^*$. This gives $\mathcal{F}'_{t} = (\dots, \epsilon_{-1}, \epsilon_{0}^*, \epsilon_{1}, \dots, \epsilon_t)$ and
\[
x'_{l,jt} = R_{l,j}(\mathcal{F}'_t).
\]

The dependence on $\epsilon_0$ can then be measured by
\[
\theta_{l,jt} = \left(\mathbb{E}\left|x_{l,jt} - x_{l,jt}' \right|^2 \right)^{1/2}
\]
where $\mathbb{E}\left| x \right|^2$ is the expected value of $\left| x \right|^2$. Controlling this dependence measure is essential for establishing the convergence rates of $\hat{f}_l$. The needed conditions on  $\theta_{l,jt}$ are laid out in Assumption~\ref{apdx_assumption_fiecas_dependence}.

\begin{assumption} \label{apdx_assumption_fiecas_dependence}
    Assume for some constant $0 < \rho < 1$, 
    \[
    \max_{l=\{1,2\}} \max_{j = 1, \dots, p} \theta_{l,jt} = O(\rho^{t}),
    \]
    and for some constant $\kappa > 0$, $C_0 > 0$,
    \[
    \max_{l=\{1,2\}} \max_{j = 1, \dots, p} \mathbb{E}\left(\exp\left( \kappa |\mathbf{x}_{l,j0}| \right) \right) \leq C_0.
    \]
\end{assumption}

Assumption~\ref{apdx_assumption_fiecas_dependence} is essentially a restatement of Assumption~1 in \citet{fiecas2019spectral}, except that we slightly generalize it to account for the fact that we have data from two separate conditions. This condition requires that the dependence on $\epsilon_0$ decreases geometrically as time increases. A similar condition is used in \citet{wu2018asymptotic} and \citet{liu2010asymptotics} and essentially requires that the dependence in the data on prior time points cannot be too strong. As noted in \citet{shao2007asymptotic} and examples 1 and 2 of \citet{liu2010asymptotics}, many processes satisfy this constraint, including ARMA, ARMA-ARCH, and ARMA-GARCH processes as well as other nonlinear autoregressive processes.

Assumption~\ref{apdx_assumption_fiecas_dependence} is satisfied if for example Assumption~1 from \citet{fiecas2019spectral} is satisfied for both conditions. For example if $\{\rho_1, \kappa_1, C_{0,1}\}$ and $\{\rho_2, \kappa_2, C_{0,2}\}$ are the parameters that satisfy Assumption 1 of \citet{fiecas2019spectral} for the data in conditions 1 and 2 respectively, then $\{\max(\rho_{1},\rho_{2}),  \min(\kappa_1,\kappa_{2}), \max(C_{0,1},C_{0,2})\}$ satisfy our Assumption~\ref{apdx_assumption_fiecas_dependence}. We are now ready to state the consistency of our estimator, $\hat{\Delta}$.

\section{THEORETICAL RESULTS} \label{apdx_proof_consistency}

To prove Theorem~\ref{thm_consistency} we will use the framework of \citet{negahban2012unified}, specifically Corollary~1. This requires establishing that the D-trace loss satisfies the Restricted Strong Convexity (RSC) condition and that the regularizing penalty is decomposable. An essential ingredient to all these results is a concentration bound on $\hat{f}_l$. We begin by establishing this bound.

\begin{lemma}[Concentration inequality on $\hat{f}_l$] \label{lemma_conc_f}
Under Assumption~\ref{apdx_assumption_fiecas_dependence}, for any $H > 0$, and for an optimal bandwidth choice of $M_{n_l}^* = O(n_l^{2/3})$, there exists a constant $C_1>0$ that depends only on $\kappa$ and $C_0$ such that 

\begin{equation}
\mathbb{P}\left( \sup_{\{\lambda_j = 2 \pi j/n_l, \lfloor (n_l-1)/2 \rfloor \leq j \leq \lfloor n_l/2 \rfloor \}} \|\hat{f}_{l}(\lambda_j) - f_{l}(\lambda_j) \|_{\infty} > C_{1} 8^{2/3}2n_l^{-\frac{1}{3}} \right) \leq C_{2} p^2 n_l^{-H} \, ,
\end{equation} 
where $l \in \{1,2\}$ indicates condition and the constant $C_2 > 0$ is constant in $n_l$ but depends on $C_1$ and $H$.
\end{lemma}

\begin{proof}

Under Assumption~\ref{apdx_assumption_fiecas_dependence} we have that Theorem 3.1 from \citep{fiecas2019spectral} holds so that for any $\delta > 0$ and $H > 0$ and for $C_1 > 0$ depending only on $\kappa$ and $C_0$

\[
    \mathbb{P}\left( \sup_{\{\lambda_j = 2 \pi j/n_l, \lfloor (n_l-1)/2 \rfloor \leq j \leq \lfloor n_l/2 \rfloor \}} \|\hat{f}_{l}(\lambda_j) - f_{l}(\lambda_j) \|_{\infty} > \frac{C_{1} M_{n_l}}{n_l} + \frac{8 n_l^{\delta}}{M_{n_l}^{1/2 + \delta}}\right) \leq C_{2} p^2 n_l^{-H} \, ,
\]

where $l$ indexes condition 1 or 2 and $C_2$ is constant in $n_l$ but depends on $C_1$ and $H$. Note that this same inequality holds for both conditions $l = 1,2$ because Assumption~\ref{apdx_assumption_fiecas_dependence} holds for both conditions. We allow $n_l$ to differ between conditions for full generality. We can then select $M_{n_l}$ to minimize the deviations. That is we select $M_{n_l}$ to minimize $\frac{C_{1} M_{n_l}}{n_l} + \frac{8 n_l^{\delta}}{M_{n_l}^{1/2 + \delta}}$.  The algebra is omitted but setting $\frac{\partial{}}{\partial{M_{n_l}}} \left[\frac{C_{1} M_{n_l}}{n_l} + \frac{8 n_l^{\delta}}{M_{n_l}^{1/2 + \delta}}\right] = 0$ and solving we get $M^{*}_{n_l} = \left(\frac{8}{C_{1}}\left(\frac{1}{2} + \delta\right)\right)^{2/(3+2\delta)} n_l^{\frac{2 + 2\delta}{3 + 2\delta}}$. It can be shown that the second derivative is positive on $\delta > 0$ so this is indeed a global minimum (on $\delta > 0$). Plugging this back into the equation yields 

\[
\frac{C_{1} M_{n_l}}{n_l} + \frac{8 n_l^{\delta}}{M_{n_l}^{1/2 + \delta}} \Bigg|_{M_{n_l} = M_{n_l}^*} = C_{1}^{\frac{1 + 2\delta}{3 + 2 \delta}} 8^{\frac{2}{3 + 2\delta}}\left( \left(\frac{1}{2} + \delta \right)^{\frac{2}{3 + 2\delta}} + \left( \frac{1}{2} + \delta \right)^{-\frac{1 + 2\delta}{3 + 2\delta}} \right) n_l^{-\frac{1}{3 + 2\delta}} \, .
\]

One can further show that $\left( \left(\frac{1}{2} + \delta \right)^{\frac{2}{3 + 2\delta}} + \left( \frac{1}{2} + \delta \right)^{-\frac{1 + 2\delta}{3 + 2\delta}} \right)$ is maximized at $\delta = 1/2$ with a maximum value of $2$ (algebra omitted). Furthermore $C_{1}^{\frac{1 + 2\delta}{3 + 2 \delta}} \leq C_{1}$ and $8^{\frac{2}{3 + 2\delta}} \leq 8^{2/3}$ for all $\delta > 0$. Therefore we have that for all $\delta > 0$,

\[
\frac{C_{1} M^*_{n_l}}{n_l} + \frac{8 n_l^{\delta}}{(M^*_{n_l})^{1/2 + \delta}} \leq C_{1}8^{2/3}2n_l^{-\frac{1}{3 + 2\delta}} \, .
\]

Letting $\delta \rightarrow 0$ the RHS of our concentration bound becomes $C_1 8^{2/3}2n_l^{-\frac{1}{3}}$ establishing the lemma.

\end{proof}

\begin{remark}
    Since $\|\hat{\Sigma}_{l}(\lambda_j) - \Sigma_{l}(\lambda_j) \|_{\infty} \leq \|\hat{f}_{l}(\lambda_j) - f_{l}(\lambda_j) \|_{\infty}$, it follows that Lemma~\ref{lemma_conc_f} also holds for $\|\hat{\Sigma}_{l}(\lambda_j) - \Sigma_{l}(\lambda_j) \|_{\infty}$.
\end{remark}

To establish RSC of the D-trace loss function we will also need a concentration inequality on the second derivative of the D-trace loss which is stated in Lemma~\ref{lemma_conc_hessian}.

\begin{lemma}[Concentration inequality on second derivative of D-trace loss] \label{lemma_conc_hessian}

Under the conditions from Lemma~\ref{lemma_conc_f} we have that for $\min(n_1,n_2) \geq 8^{9/2}C_{1}^3/C_f^3$, with probability $> 1 - 2C_2p^2\min(n_1,n_2)^{-H}$, 

\begin{equation*}
\left\|0.5 \left( \hat{\Sigma}_1 \otimes \hat{\Sigma}_2 + \hat{\Sigma}_2 \otimes \hat{\Sigma}_1 \right) - 0.5\left(\Sigma_1 \otimes \Sigma_2  + \Sigma_2 \otimes \Sigma_1 \right)\right\|_{\infty} \leq \frac{C_fC_{1}8^{5/2}}{\min(n_1,n_2)^{1/3}} \, ,
\end{equation*}

where $C_{f} = \max(\|f_{1}\|_{\infty}, \|f_{2}\|_{\infty})$ and $C_1$ is the same as in Lemma~\ref{lemma_conc_f}.
\end{lemma}

\begin{proof}

Begin by noting that both 
\[
\|\hat{\Sigma}_2 \otimes \hat{\Sigma}_1 - \Sigma_2 \otimes \Sigma_1 \|_{\infty}, \|\hat{\Sigma}_1 \otimes \hat{\Sigma}_2 - \Sigma_1 \otimes \Sigma_2 \|_{\infty} \leq \|\Sigma_1\|_{\infty} \|\hat{\Sigma}_2 - \Sigma_2\|_{\infty} + \|\Sigma_2\|_{\infty} \|\hat{\Sigma}_1 - \Sigma_1\|_{\infty} + \|\hat{\Sigma}_1 - \Sigma_1\|_{\infty}\|\hat{\Sigma}_2 - \Sigma_2\|_{\infty} \, .
\]

Also we have that $\|\Sigma_l \|_{\infty} \leq \|f_{l}\|_{\infty}$. Using Lemma~\ref{lemma_conc_f}, we get that with probability $> 1 - 2C_{2} p^2\min(n_1, n_2)^{-H}$,
\begin{align*}
\|\hat{\Sigma}_1 \otimes \hat{\Sigma}_2 - \Sigma_1 \otimes \Sigma_2 \|_{\infty} & \leq C_f\left( C_{1}8^{3/2}2n_1^{-1/3} + C_{1}8^{3/2}2n_2^{-1/3} \right) + C_{1}^2 8^{6/2}4n_1^{-1/3}n_2^{-1/3} \\
& \leq C_fC_{1}8^{3/2} 4 \min(n_1, n_2)^{-1/3} + C_{1}^2 8^{6/2}4\min(n_1, n_2)^{-2/3} \, .
\end{align*}

We can see that for $\min(n_1,n_2) \geq 8^{9/2}C_{1}^3/C_f^3 $ the 2nd term on the RHS is smaller than the 1st term on the RHS so we can combine them to simplify the bound to be 2 times the first term on the RHS or

\[
\|\hat{\Sigma}_1 \otimes \hat{\Sigma}_2 - \Sigma_1 \otimes \Sigma_2 \|_{\infty} \leq \frac{C_fC_{1}8^{5/2}}{\min(n_1,n_2)^{1/3}}.\]

This also holds for $\|\hat{\Sigma}_2 \otimes \hat{\Sigma}_1 - \Sigma_2 \otimes \Sigma_1 \|_{\infty}$. Defining 

\[
\text{Event } \mathcal{A} := \left\{ \|\hat{\Sigma}_1 \otimes \hat{\Sigma}_2 - \Sigma_1 \otimes \Sigma_2 \|_{\infty} \leq \frac{C_fC_{1}8^{5/2}}{\min(n_1,n_2)^{1/3}} \right\} \cap \left\{ \|\hat{\Sigma}_2 \otimes \hat{\Sigma}_1 - \Sigma_2 \otimes \Sigma_1 \|_{\infty} \leq \frac{C_fC_{1}8^{5/2}}{\min(n_1,n_2)^{1/3}} \right\}
\]

then from the above results, for $\min(n_1,n_2) \geq 8^{9/2}C_{1}^3/C_f^3$,  Event $\mathcal{A}$ occurs with probability $> 1 - 2C_{2}p^2\min(n_1,n_2)^{-H}$ and we have

\begin{equation}
\left\|0.5 \left( \hat{\Sigma}_1 \otimes \hat{\Sigma}_2 + \hat{\Sigma}_2 \otimes \hat{\Sigma}_1 \right) - 0.5\left(\Sigma_1 \otimes \Sigma_2  + \Sigma_2 \otimes \Sigma_1 \right)\right\|_{\infty} \leq \frac{C_fC_{1}8^{5/2}}{\min(n_1,n_2)^{1/3}} \, .
\end{equation}

\end{proof}

Using Lemma~\ref{lemma_conc_f} and Lemma~\ref{lemma_conc_hessian} we can establish that the D-trace loss satisfies the restricted strong convexity condition. This is formalized in Lemma~\ref{lemma_rsc_dtrace}

\begin{lemma}[RSC of D-trace loss] \label{lemma_rsc_dtrace}

Under the conditions of Lemma~\ref{lemma_conc_f}, for $\min(n_1,n_2) \geq \max\left(8^{9/2}C_{1}^3/C_f^3, \left( \frac{32 s_{\Delta} C_f C_{1} 8^{5/2}}{\lambda_{\mathrm{min}}(f_{1})\lambda_{\mathrm{min}}(f_{2})}\right)^3 \right)$  with probability > $1 - 2C_{2}p^2\min(n_1,n_2)^{-H}$  RSC holds for the D-trace loss with $\kappa_L = \lambda_{\mathrm{min}}(f_{1})\lambda_{\mathrm{min}}(f_{2})/2$. We use $\lambda_{\mathrm{min}}(A)$ to denote the minimum eigenvalue of a matrix $A$ and $s_{\Delta}$ to denote the number of non-zero entries of $\Delta$. The constant $C_f$ is defined in Lemma~\ref{lemma_conc_hessian} while $C_1, C_2$ are defined in Lemma~\ref{lemma_conc_f}.

\end{lemma}

\begin{proof}

To show that the D-trace loss function satisfies the RSC condition, we must show that $\mathbf{m}^T \left( \nabla^2 L(\Delta, \hat{\Sigma}_1, \hat{\Sigma}_2) \right) \mathbf{m} > \kappa_L \|\mathbf{m}\|_2^2$ for all $\mathbf{m} \in \mathcal{C}(S_{\Delta})$ where $\mathcal{C}(S_{\Delta}) = \left\{ \mathbf{m} \in \mathbb{R}^{4p^2} \big| \|\mathbf{m}_{S_{\Delta}^c}\|_1 \leq 3 \|\mathbf{m}_{S_{\Delta}}\|_1 \right\}$ for some $\kappa_L$. For our loss, we let $S_{\Delta}$ be the support of $\Delta$, the true difference. The notation $\mathbf{m}_{S_{\Delta}}$ selects the entries of $\mathbf{m}$ that correspond to the non-zero indices of $\Delta$ while $\mathbf{m}_{S_{\Delta}^c}$ selects the entries corresponding to the zero indices of $\Delta$.

First note that for $\mathbf{m} \in \mathcal{C}(S_{\Delta})$, $\|\mathbf{m}\|_1 \leq 4 \|\mathbf{m}_{S_{\Delta}}\|_1 \leq 4\sqrt{s_{\Delta}}\|\mathbf{m}_{S_{\Delta}}\|_2$. Recall $\nabla^{2}L(\Delta, \hat{\Sigma}_1, \hat{\Sigma}_2) = 0.5(\hat{\Sigma}_1 \otimes \hat{\Sigma}_2 + \hat{\Sigma_2} \otimes \hat{\Sigma}_1)$. Showing RSC proceeds the same as in \citet{wang2021direct}.

\begin{align}
\begin{split}
\label{RSC_ineq}
\mathbf{m}^T (\hat{\Sigma}_1 \otimes \hat{\Sigma}_2) \mathbf{m} & \geq \mathbf{m}^T (\Sigma_1 \otimes \Sigma_2) \mathbf{m} + \mathbf{m}^T \left( \hat{\Sigma}_1 \otimes \hat{\Sigma}_2 - \Sigma_1 \otimes \Sigma_2 \right) \mathbf{m} \\
& \geq \lambda_{\mathrm{min}}(\Sigma_1) \lambda_{\mathrm{min}}(\Sigma_2) \|\mathbf{m}\|_2^2 - \left| \mathbf{m}^T \left(\hat{\Sigma}_1 \otimes \hat{\Sigma}_2 - \Sigma_1 \otimes \Sigma_2  \right) \mathbf{m} \right| \\
& \geq \lambda_{\mathrm{min}}(\Sigma_1) \lambda_{\mathrm{min}}(\Sigma_2) \|\mathbf{m}\|_2^2 - \|\hat{\Sigma}_1 \otimes \hat{\Sigma}_2 - \Sigma_1 \otimes \Sigma_2 \|_{\infty} \|\mathbf{m}\|_1^{2} \\
& \geq \lambda_{\mathrm{min}}(\Sigma_1) \lambda_{\mathrm{min}}(\Sigma_2) \|\mathbf{m}\|_2^2 - 16s_{\Delta}\|\hat{\Sigma}_1 \otimes \hat{\Sigma}_2 - \Sigma_1 \otimes \Sigma_2 \|_{\infty} \|\mathbf{m}\|_2^{2} \, .
\end{split}
\end{align}

The last inequality follows from the fact that $\| \mathbf{m} \|_{1} = \|\mathbf{m}_{S_{\Delta}^c}\|_1 + \|\mathbf{m}_{S_{\Delta}}\|_1 \leq 4\|\mathbf{m}_{S_{\Delta}}\|_1 \leq 4\sqrt{s_{\Delta}}\|\mathbf{m}_{S_{\Delta}}\|_2 \leq 4\sqrt{s_{\Delta}}\|\mathbf{m}\|_2$. Similarly

\begin{equation*} 
\mathbf{m}^T (\hat{\Sigma}_2 \otimes \hat{\Sigma}_1) \mathbf{m} \geq\lambda_{\mathrm{min}}(\Sigma_1) \lambda_{\mathrm{min}}(\Sigma_2) \|\mathbf{m}\|_2^2 - 16s_{\Delta}\|\hat{\Sigma}_2 \otimes \hat{\Sigma}_1 - \Sigma_2 \otimes \Sigma_1 \|_{\infty} \|\mathbf{m}\|_2^{2} \, .
\end{equation*}

Using Lemma~\ref{lemma_conc_hessian} and the fact that $\lambda_{\mathrm{min}}(\Sigma_l) = \lambda_{\mathrm{min}}(f_{l})$  we have that, for $\min(n_1,n_2) \geq 8^{9/2}C_{1}^3/C_f^3$ and with probability $> 1 - 2C_2p^2\min(n_1,n_2)^{-H}$, 

\[
\mathbf{m}^T 0.5(\hat{\Sigma}_1 \otimes \hat{\Sigma}_2 + \hat{\Sigma}_2 \otimes \hat{\Sigma}_1) \mathbf{m} \geq \lambda_{\mathrm{min}}(f_{1}) \lambda_{\mathrm{min}}(f_{2}) \|\mathbf{m}\|_2^2 - 16s_{\Delta} \frac{C_fC_{1}8^{5/2}}{\min(n_1,n_2)^{1/3}} \|\mathbf{m}\|_2^{2}
\]

So for $\min(n_1,n_2) \geq \left( \frac{32 s_{\Delta} C_f C_{1} 8^{5/2}}{\lambda_{\mathrm{min}}(f_{1})\lambda_{\mathrm{min}}(f_{2})}\right)^3$ then the 2nd term on the RHS is $\leq$ 0.5*(1st term on RHS). Thus we get

\[
\mathbf{m}^T 0.5(\hat{\Sigma}_1 \otimes \hat{\Sigma}_2 + \hat{\Sigma}_2 \otimes \hat{\Sigma}_1) \mathbf{m} \geq \frac{\lambda_{\mathrm{min}}(f_{1}) \lambda_{\mathrm{min}}(f_{2})}{2} \|\mathbf{m}\|_2^2
\]

Therefore we conclude that for $\min(n_1,n_2) \geq \max\left(8^{9/2}C_{1}^3/C_f^3, \left( \frac{32 s_{\Delta} C_f C_{1} 8^{5/2}}{\lambda_{\mathrm{min}}(f_{1})\lambda_{\mathrm{min}}(f_{2})}\right)^3 \right)$  with probability > $1 - 2C_{2}p^2\min(n_1,n_2)^{-H}$  RSC holds with $\kappa_L = \lambda_{\mathrm{min}}(f_{1})\lambda_{\mathrm{min}}(f_{2})/2$.

\end{proof}

We are now ready to prove Theorem~\ref{thm_consistency}.

\begin{proof}[Proof of Theorem~\ref{thm_consistency}]

    We first proceed by establishing the decomposability of our penalty term, $\left\| \Delta \right\|_1$. Decomposability is defined in Definition 1 of \citet{negahban2012unified}. Let $\overline{\mathcal{M}} = \left\{ \theta : \mathbb{R}^{2p \times 2p} \big| \theta_{ij} \neq 0 \text{ for all } (i,j) \in S_{\Delta} \right\}$ and $\overline{\mathcal{M}}^{\perp} = \left\{ \gamma : \mathbb{R}^{2p \times 2p} \big| \gamma_{ij} = 0 \text{ for all } (i,j) \in S_{\Delta} \right\}$. Then $\| \cdot \|_1$ is decomposable since $\|\theta + \gamma\|_1 = \|\theta \|_1 + \|\gamma\|_1$ for all $\theta \in \overline{\mathcal{M}}$ and $\gamma \in \overline{\mathcal{M}}^{\perp}$.

    Next we discuss a suitable choice of penalty scale $\tau_{n_1, n_2}$. From Theorem 1 in \citet{negahban2012unified} we require $\tau_{n_1, n_2} \geq 2\|\nabla L_D(\Delta, \hat{\Sigma}_1, \hat{\Sigma}_2)\|_{\infty}$. We know that $2 \| \nabla L_D(\Delta, \hat{\Sigma}_1, \hat{\Sigma}_2)\|_{\infty} = 2 \|0.5\left( \hat{\Sigma}_1 \Delta \hat{\Sigma}_2 + \hat{\Sigma}_2 \Delta \hat{\Sigma}_1 \right) - \left( \hat{\Sigma}_1 - \hat{\Sigma}_2 \right)\|_{\infty}.$
    
    Let $\Gamma = 0.5\left(\Sigma_2 \otimes \Sigma_1 + \Sigma_1 \otimes \Sigma_2 \right)$ and $\hat{\Gamma} = 0.5\left(\hat{\Sigma}_2 \otimes \hat{\Sigma}_1 + \hat{\Sigma}_1 \otimes \hat{\Sigma}_2 \right).$ Then we write
    \begin{align*}
    \|0.5\left( \hat{\Sigma}_1 \Delta \hat{\Sigma}_2 + \hat{\Sigma}_2 \Delta \hat{\Sigma}_1 \right) - \left( \hat{\Sigma}_1 - \hat{\Sigma}_2 \right)\|_{\infty} & = \|\hat{\Gamma} \text{vec}(\Delta) - \left( \text{vec}(\hat{\Sigma}_1) - \text{vec}(\hat{\Sigma}_2)\right) \|_{\infty} \\
    & = \|(\hat{\Gamma} - \Gamma)\text{vec}(\Delta) - \left( \text{vec}(\hat{\Sigma}_1) - \text{vec}(\Sigma_1) \right) - \left( \text{vec}(\hat{\Sigma}_2) - \text{vec}(\Sigma_2)  \right)\|_{\infty} \\
    & \leq \|(\hat{\Gamma} - \Gamma)\text{vec}(\Delta)\|_{\infty} + \|\hat{\Sigma}_1 - \Sigma_1 \|_{\infty} +\|\hat{\Sigma}_2 - \Sigma_2 \|_{\infty} \\
    & \leq \|\hat{\Gamma} - \Gamma\|_{\infty} \|\Delta\|_1 + \|\hat{\Sigma}_1 - \Sigma_1 \|_{\infty} + \|\hat{\Sigma}_2 - \Sigma_2 \|_{\infty} \, .\\
    \end{align*}

    By Lemma~\ref{lemma_conc_hessian} we have for $\min(n_1, n_2) \geq \frac{C_{1}^3}{C_f^3} 8^{9/2}$, with probability $> 1 - 2C_{2}p^2\min(n_1,n_2)^{-H}$,
    \begin{align}
    \label{eqn_dtrace_deriv_norm}
    \left\|0.5\left( \hat{\Sigma}_1 \Delta \hat{\Sigma}_2 + \hat{\Sigma}_2 \Delta \hat{\Sigma}_1 \right) - \left( \hat{\Sigma}_1 - \hat{\Sigma}_2 \right)\right\|_{\infty} & \leq \frac{C_{1} C_f 8^{5/2}}{\min(n_1,n_2)^{1/3}} \|\Delta\|_1 + \frac{C_{1}8^{3/2}2}{n_1^{1/3}} + \frac{C_{1}8^{3/2}2}{n_2^{1/3}} \nonumber \\
    & \leq \frac{\left(C_{1} C_f 8^{5/2} \|\Delta\|_1 + 4C_{1}8^{3/2}\right)}{\min(n_1,n_2)^{1/3}} \, .
    \end{align}
    
    Therefore if we choose 
    \[
    \tau_{n_1, n_2} \geq \frac{2\left(C_{1} C_f 8^{5/2} \|\Delta\|_1 + 4C_{1}8^{3/2}\right)}{\min(n_1,n_2)^{1/3}} \, ,
    \]
    we will satisfy $\tau_{n_1, n_2} \geq 2\|\nabla L_D(\Delta, \hat{\Sigma}_1, \hat{\Sigma}_2)\|_{\infty}$ with probability $> 1 - 2C_{2}p^2\min(n_1,n_2)^{-H}$ when $\min(n_1, n_2) \geq \frac{C_{1}^3}{C_f^3} 8^{9/2}$.

    Combining the above with Lemma~\ref{lemma_rsc_dtrace} we get that for $\tau_{n_1, n_2} \geq \frac{2}{\min(n_1,n_2)^{1/3}}\left(C_{1} C_f 8^{5/2} \|\Delta\|_1 + 4C_{1}8^{3/2}\right)$, $\min(n_1,n_2) \geq \max\left(8^{9/2}C_{1}^3/C_f^3, \left( \frac{32 s_{\Delta} C_f C_{1} 8^{5/2}}{\lambda_{\mathrm{min}}(f_{1})\lambda_{\mathrm{min}}(f_{2})}\right)^3 \right)$, and $M_{n,l} = O(n_l^{2/3})$ we have with probability greater than $1 - Cp^2\min(n_1,n_2)^{-H}$ for any $H > 0$, 
    
    \[
    \| \hat{\Delta} - \Delta \|_{F} \leq \frac{\sqrt{18 s_{\Delta}} \tau_{n_1, n_2} }{\lambda_{\mathrm{min}}(f_{1})\lambda_{\mathrm{min}}(f_{2})} \, ,
    \]
    where we have defined $C = 2C_2$. Since $\tau_{n_1, n_2}$ is of order $\min(n_1,n_2)^{-1/3}$, convergence of our D-trace estimator is also of order $\min(n_1,n_2)^{-1/3}$.

\end{proof}

\section{SIMULATION SPECIFICS AND RESULTS} \label{apdx_sim}

We generate data in condition $l$ as
\[
\mathbf{x}_{l,t} = A_l \mathbf{x}_{l,t-1} + \boldsymbol{\epsilon}_{l,t} \,
\]
where $\boldsymbol{\epsilon}_{l,t} \sim N_{p}\left(\mathbf{0}, I_p \right)$. From \citet{sun2018large}, the spectral density at frequency $\lambda$ is known to be
\[
f_l(\lambda) = \frac{1}{2\pi}\left( \mathcal{A}_l(e^{-i\lambda}) \right)^{-1} I_p \left(\left( \mathcal{A}_l(e^{-i\lambda}) \right)^{-1}\right)^{H}, 
\]
where $\mathcal{A}_l(z) =I_p - A_l z$. When the transition matrix $A_l$ is block diagonal, the spectral density $f_l(\lambda)$ is also block diagonal. Since the inverse of a block diagonal matrix can be computed block by block we can easily generate a sparse difference matrix by enforcing the transition matrix in conditions $1$ and $2$ to be the same except for a small block. For example, with $p = 54$ if we generate $A_1$ and $A_2$ as block diagonal with the same $51 \times 51$ block and only differing in the final $3 \times 3$ block, then their spectral densities will only differ in this last $3 \times 3$ block. When converting to the real space, the expanded $\Sigma_i \in (2p \times 2p)$ and thus the true difference $\Delta = \Sigma_1 - \Sigma_2$ will differ in four $3 \times 3$ blocks.

For simulation setting 1, the transition matrix is the same as in \citep{sun2018large}. Specifically, it consists of $18$ blocks of dimension $3 \times 3$, where each block is $\begin{bmatrix} 0.5 & 0.9 & 0 \\ 0 & 0.5 & 0.9 \\ 0 & 0 & 0.5 \end{bmatrix}$.  In both simulation setting 2 and setting 3, 60\% of the coefficients in the $3 \times 3$ block were randomly drawn from either a $\mathrm{Uniform}(-0.5,-0.2)$ or a $\mathrm{Uniform}(0.2,0.5)$ each with equal probability. In the second and third setting, 40\% and 5\% respectively of the entries of the larger block were randomly drawn from a $\mathrm{Uniform}(-0.5,-0.2)$ and $\mathrm{Uniform}(0.2,0.5)$ each with equal probability. This corresponded to 60\% and 95\% sparsity respectively. The number of non-zero entries in the difference in expanded inverse spectral densities, which we will refer to as edges, varies by frequency but is almost always 22 for Simulation 1, 14 for Simulation 2, and 28 for Simulation 3.

For all settings $100$ evenly spaced Fourier frequencies from $0$ to $\pi - n^{-1}$ was used. These were used as the spectral density is conjugate symmetric around $0$. In the case of $n = 100$ only $50$ Fourier frequencies were used as there are only 50 Fourier frequencies from $0$ to $\pi - n^{-1}$. Given the availability of results from multiple frequencies, each simulation was run once and the results are averaged across all frequencies. The standard error of each metric was also computed across frequencies and are reported in parentheses.

To solve the SDD estimation problem (\ref{eqn_dtrace_est}), the alternating direction method of multipliers (ADMM) algorithm from \citet{yuan2017differential} was used. Specifically, we used the \texttt{L1\_dts} function in the \texttt{Difdtl} package available on GitHub at SusanYuan/Difdtl (GPL ($\geq$ 2)). To generate the sequence of penalties, $\{ \tau_{n_1, n_2} \}$, we used 20 values on a log-linear scale from $0.001*\tau_{n_1, n_2, \mathrm{max}}$ to  $\tau_{n_1, n_2, \mathrm{max}}$ where  $\tau_{n_1, n_2, \mathrm{max}}$ represents the minimum value where all entries of $\hat{\Delta}$ are 0. In this case,  $\tau_{n_1, n_2, \mathrm{max}} = 2\max(|\hat{\Sigma}_1 - \hat{\Sigma}_2|)$. The penalty $\tau_{n_1,n_2}^*$ was chosen as the $\tau_{n_1, n_2}$ that minimizes the eBIC which is computed as

\begin{equation}
\label{eqn_sdd_ebic}
    \mathrm{eBIC}(\hat{\Delta})_{\gamma} = \min(n_1,n_2)\left\| \frac{1}{2} \left( \hat{\Sigma}_1 \hat{\Delta} \hat{\Sigma}_2 + \hat{\Sigma}_2 \hat{\Delta} \hat{\Sigma}_1  - \hat{\Sigma}_2 + \hat{\Sigma}_1 \right) \right\|_{\infty} + \log(\min(n_1,n_2))|E| + 4 \gamma |E| \log(p) \, ,
\end{equation}

where $|E|$ is the number of unique edges in $\hat{\Delta}$. Since $\hat{\Delta}$ represents the difference in expanded spectral densities $|E|$ is the number of non-zero entries in the upper triangular portions, including diagonals, of the submatrices $\hat{\Delta}_{1:p,1:p}$ and $\hat{\Delta}_{1:p,(p+1):2p}$. For all applications we use $\gamma = 0.5$. For FGL, Na\"{i}ve, and Hard Thresholding, 20 tuning parameter values were also used. FGL requires two tuning parameter values, $\lambda_1, \lambda_2,$ so a grid of 20 $(\lambda_1, \lambda_2)$ pairs was considered. Since the parameter $\lambda_1$ controls the sparsity of each inverse matrix while $\lambda_2$ controls sparsity of the difference we considered only two values for $\lambda_1$ and 10 values for $\lambda_2$. The values for $\lambda_1$ were $(0.01\lambda_{1,\max}, 0.1\lambda_{1,\max})$ where $\lambda_{1,\max} = \left\| \hat{\Sigma}_{1,O} + \hat{\Sigma}_{2,O} \right\|_{\infty}/2$ and $\hat{\Sigma}_{i,O}$ is the off-diagonals of $\hat{\Sigma}_i$. The $\lambda_2$ sequence was generated using 10 values on a log-linear scale from $0.0001* \lambda_{2,\max}$ to $\lambda_{2,\max}$ where $\lambda_{2,\max} = \max\left( \left\| \hat{\Sigma}_{1,O} \right\|_{\infty} -  0.01\lambda_{1,\max}, \left\| \hat{\Sigma}_{2,O} \right\|_{\infty} -  0.01\lambda_{1,\max} \right)$. For the Na\"{i}ve method, we first computed $\lambda_{\max, i} = \left\| \hat{\Sigma}_i \right\|_{\infty}$ in each condition and then computed the tuning parameters as 20 evenly-spaced values on a log-linear scale from $0.0001 \lambda_{\max, i}$ to $\lambda_{\max,i}$. Similarly, for the Hard Thresholding method, we used 20 evenly-spaced values on a log-linear scale from $\min(|\hat{\Delta}_H|)$ to $\max(|\hat{\Delta}_H|)$ where $\hat{\Delta}_H$ is generated by inverting $\hat{\Sigma}_i$ and forming the difference. Note that for $n = 100$, it is not possible to invert $\hat{\Sigma}_i$.

Let TP, FN, TN, FP denote the true positive, false negative, true negative, and false positive edges identified by either SDD or the Na\"ive method, respectively. A value greater than $1\times10^{-6}$ in absolute value was considered an edge. The metrics are defined as follows
\begin{align*}
\mathrm{Precision} = \frac{\mathrm{TP}}{\mathrm{TP} + \mathrm{FP}} & \qquad & \mathrm{Recall} = \frac{\mathrm{TP}}{\mathrm{TP} + \mathrm{FN}} \\
\mathrm{Accuracy} = \frac{\mathrm{TP} + \mathrm{TN}}{4p^2} & \qquad & \mathrm{RRMSE} = \sqrt{\frac{ \sum_{i,j} (\hat{\Delta}_{i,j} - \Delta_{i,j})^2}{\sum_{i,j} (\Delta_{i,j})^2}}, 
\end{align*} 

where $\hat{\Delta}$ represents the difference estimator, either SDD or the Na\"ive difference, $\Delta$ is the true difference. The denominator of the accuracy measure is $4p^2$ as expanding a $p \times p$ spectral density to the real space gives a $2p \times 2p$ matrix which has $4p^2$ entries. We also report the number of average number of true edges and the average number of estimated edges across frequencies.

\begin{figure}[ht!]
    \centering
      \includegraphics[width=0.5\linewidth]{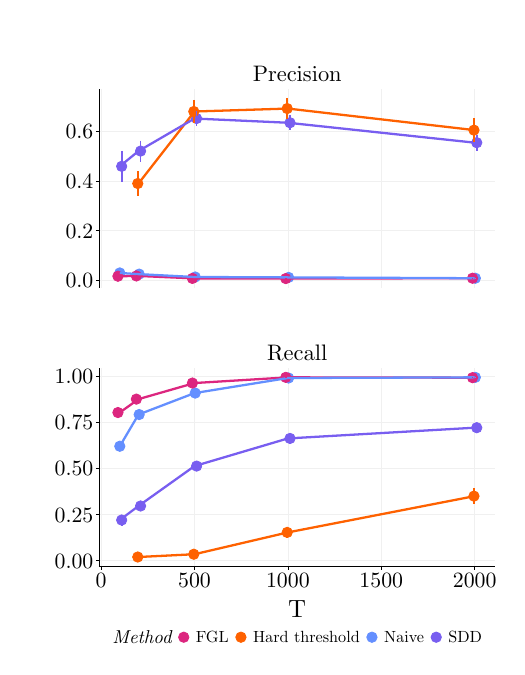} 
    \caption{Simulation 1 Precision and Recall. Results are reported as mean (dots) and SE (vertical lines) where the mean and SE are taken across all frequencies for a given sample size $n$.}
    \label{fig_sim1_prec_reca}
\end{figure}


\begin{figure}[ht!]
    \centering
      \includegraphics[width=0.5\linewidth]{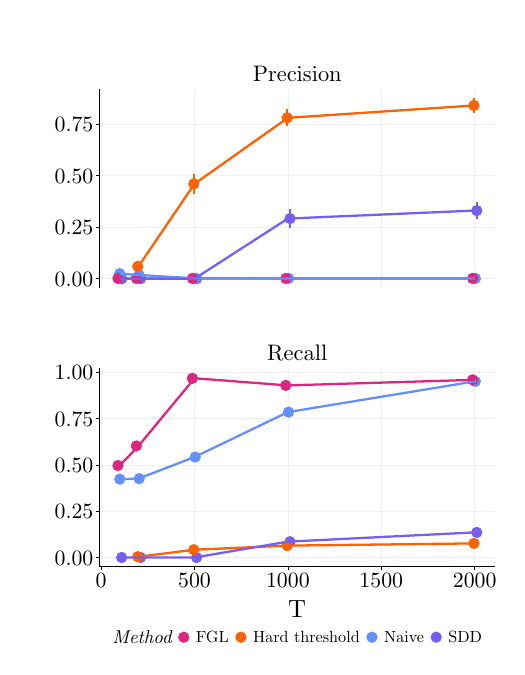} 
    \caption{Simulation 2 Precision and Recall. Results are reported as mean (dots) and SE (vertical lines) where the mean and SE are taken across all frequencies for a given sample size $n$.}
    \label{fig_sim2_prec_reca}
\end{figure}


\begin{figure}[ht!]
    \centering
      \includegraphics[width=0.5\linewidth]{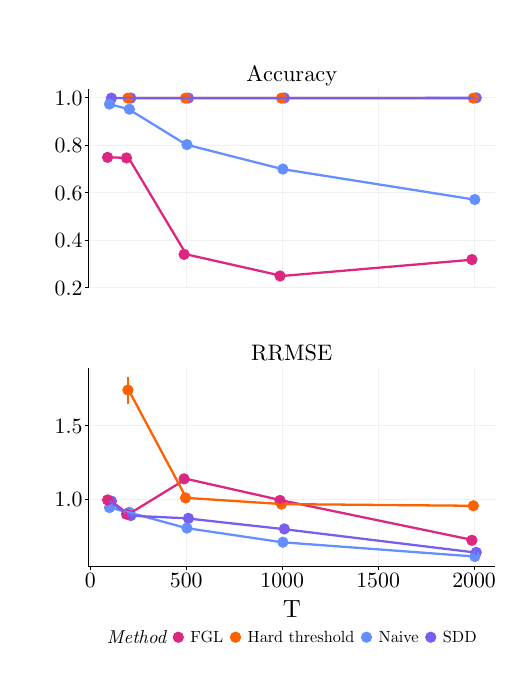} 
    \caption{Simulation 3 Accuracy and RRMSE. Results are reported as mean (dots) and SE (vertical lines) where the mean and SE are taken across all frequencies for a given sample size $n$.}
    \label{fig_sim3_acc_rrmse}
\end{figure}

\begin{figure}[ht!]
    \centering
      \includegraphics[width=0.5\linewidth]{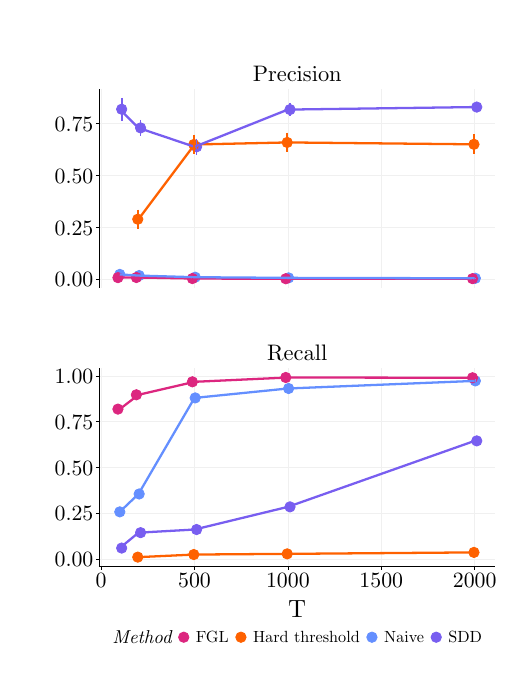} 
    \caption{Simulation 3 Precision and Recall. Results are reported as mean (dots) and SE (vertical lines) where the mean and SE are taken across all frequencies for a given sample size $n$.}
    \label{fig_sim3_prec_reca}
\end{figure}


\begin{table}[t]
  \caption{Simulation 1. Results are reported as Mean (SE) where the mean and SE are taken across all frequencies for a given sample size $n$.}
  \label{tbl_sim1}
  \centering
    \begin{tabular}{ccccccc}
        \toprule
\multicolumn{2}{c}{ } & \multicolumn{5}{c}{SDD} \\
\cmidrule(l{3pt}r{3pt}){3-7}
n & \# True edges & \# Est edges & Precision & Recall & Accuracy & RRMSE \\
\midrule
100 & 21.6 (0.32) & 10.6 (1.04) & 0.46 (0.06) & 0.22 (0.03) & 1.00 (0.00) & 0.95 (0.01)\\
200 & 21.8 (0.16) & 16.7 (1.00) & 0.52 (0.04) & 0.30 (0.02) & 1.00 (0.00) & 0.89 (0.01)\\
500 & 21.9 (0.08) & 21.1 (1.15) & 0.65 (0.03) & 0.51 (0.02) & 1.00 (0.00) & 0.78 (0.01)\\
1000 & 21.9 (0.08) & 31.1 (2.03) & 0.63 (0.03) & 0.66 (0.02) & 1.00 (0.00) & 0.61 (0.01)\\
2000 & 21.9 (0.08) & 48.7 (4.60) & 0.55 (0.03) & 0.72 (0.02) & 1.00 (0.00) & 0.53 (0.01)\\
&   &   &   &   &   &  \\

\multicolumn{2}{c}{ } & \multicolumn{5}{c}{FGL} \\
\cmidrule(l{3pt}r{3pt}){3-7}
 & & 2073 (194) & 0.02 (0.00) & 0.80 (0.02) & 0.82 (0.02) & 0.89 (0.01)\\
 & & 2444 (137) & 0.02 (0.00) & 0.88 (0.01) & 0.79 (0.01) & 0.84 (0.01)\\
 &  & 4973 (326) & 0.01 (0.00) & 0.96 (0.01) & 0.58 (0.03) & 0.68 (0.01)\\
 &  & 5962 (324) & 0.01 (0.00) & 0.99 (0.00) & 0.49 (0.03) & 0.64 (0.02)\\
 &  & 6308 (327) & 0.01 (0.00) & 0.99 (0.00) & 0.46 (0.03) & 0.62 (0.02)\\
 &   &   &   &   &   &  \\

\multicolumn{2}{c}{ }& \multicolumn{5}{c}{Na\"{i}ve} \\
\cmidrule(l{3pt}r{3pt}){3-7} 
 &  & 659 (46.1) & 0.03 (0.00) & 0.62 (0.02) & 0.94 (0.00) & 0.95 (0.01)\\
 &  & 1077 (56.9) & 0.03 (0.00) & 0.79 (0.02) & 0.91 (0.00) & 0.90 (0.00)\\
 &  & 1775 (72.1) & 0.01 (0.00) & 0.91 (0.01) & 0.85 (0.01) & 0.80 (0.01)\\
 &  & 2289 (104) & 0.01 (0.00) & 0.99 (0.00) & 0.81 (0.01) & 0.69 (0.01)\\
 &  & 3158 (151) & 0.01 (0.00) & 0.99 (0.00) & 0.73 (0.01) & 0.58 (0.01)\\
 &  &   &   &   &   &  \\
 
 \multicolumn{2}{c}{ } & \multicolumn{5}{c}{Hard threshold } \\
\cmidrule(l{3pt}r{3pt}){3-7} 
 &  & - & - & - & - & -\\
 &  & 0.69 (0.05) & 0.39 (0.05) & 0.02 (0.00) & 1.00 (0.00) & 1.57 (0.06)\\
 &  & 0.79 (0.06) & 0.68 (0.05) & 0.03 (0.00) & 1.00 (0.00) & 1.00 (0.00)\\
 &  & 5.01 (1.11) & 0.69 (0.04) & 0.15 (0.03) & 1.00 (0.00) & 0.89 (0.02)\\
 &  & 230 (53.1) & 0.60 (0.05) & 0.35 (0.04) & 0.98 (0.00) & 1.04 (0.06)\\
        \bottomrule
        \end{tabular}
\end{table}

\begin{table}[t]
  \caption{Simulation 2. Results are reported as Mean (SE) where the mean and SE are taken across all frequencies for a given sample size $n$.}
  \label{tbl_sim2}
  \centering
\begin{tabular}{ccccccc}
        \toprule
\multicolumn{2}{c}{ } & \multicolumn{5}{c}{SDD} \\
\cmidrule(l{3pt}r{3pt}){3-7}
n & \# True edges & \# Est edges & Precision & Recall & Accuracy & RRMSE \\
\midrule
100 & 13.7 (0.21) & 1.52 (0.27) & 0.00 (0.00) & 0.00 (0.00) & 1.00 (0.00) & 1.00 (0.00)\\
200 & 13.9 (0.11) & 0.90 (0.13) & 0.00 (0.00) & 0.00 (0.00) & 1.00 (0.00) & 1.00 (0.00)\\
500 & 14.0 (0.04) & 2.24 (0.35) & 0.00 (0.00) & 0.00 (0.00) & 1.00 (0.00) & 1.00 (0.00)\\
1000 & 14.0 (0.04) & 3.80 (0.47) & 0.29 (0.04) & 0.09 (0.01) & 1.00 (0.00) & 0.99 (0.01)\\
2000 & 14.0 (0.04) & 6.52 (0.83) & 0.33 (0.04) & 0.14 (0.01) & 1.00 (0.00) & 0.95 (0.01)\\
&   &   &   &   &   &   \\

\multicolumn{2}{c}{ } & \multicolumn{5}{c}{FGL} \\
\cmidrule(l{3pt}r{3pt}){3-7} 
 &  & 4675 (227) & 0.00 (0.00) & 0.50 (0.02) & 0.60 (0.02) & 2.25 (0.16)\\
 &  & 5297 (185) & 0.00 (0.00) & 0.60 (0.02) & 0.55 (0.02) & 2.15 (0.12)\\
 &  & 10201 (118) & 0.00 (0.00) & 0.97 (0.01) & 0.13 (0.01) & 4.20 (0.07)\\
 &  & 9731 (155) & 0.00 (0.00) & 0.93 (0.02) & 0.17 (0.01) & 3.09 (0.06)\\
 &  & 9196 (176) & 0.00 (0.00) & 0.96 (0.01) & 0.21 (0.02) & 2.30 (0.05)\\
&   &   &   &   &   &    \\

\multicolumn{2}{c}{ } & \multicolumn{5}{c}{Na\"{i}ve } \\
\cmidrule(l{3pt}r{3pt}){3-7} 
 & & 259 (10.9) & 0.02 (0.00) & 0.42 (0.01) & 0.98 (0.00) & 1.02 (0.01)\\
 & & 609 (46.3) & 0.02 (0.00) & 0.43 (0.00) & 0.95 (0.00) & 1.01 (0.01)\\
 & & 3566 (56.3) & 0.00 (0.00) & 0.54 (0.01) & 0.69 (0.00) & 1.20 (0.03)\\
 & & 6487 (58.2) & 0.00 (0.00) & 0.79 (0.02) & 0.44 (0.00) & 1.58 (0.04)\\
 & & 9145 (98.2) & 0.00 (0.00) & 0.95 (0.01) & 0.22 (0.01) & 2.29 (0.04)\\
 &   &   &   &   &   &  \\
 
 \multicolumn{2}{c}{ } & \multicolumn{5}{c}{Hard threshold } \\
\cmidrule(l{3pt}r{3pt}){3-7} 
 & & - & - & - & - & -\\
 &  & 0.72 (0.05) & 0.06 (0.02) & 0.00 (0.00) & 1.00 (0.00) & 3.05 (0.18)\\
 &  & 0.93 (0.08) & 0.46 (0.05) & 0.04 (0.01) & 1.00 (0.00) & 1.09 (0.02)\\
 &  & 0.93 (0.06) & 0.78 (0.04) & 0.06 (0.00) & 1.00 (0.00) & 0.94 (0.01)\\
 &  & 1.08 (0.06) & 0.84 (0.04) & 0.08 (0.00) & 1.00 (0.00) & 0.89 (0.01)\\
        \bottomrule
\end{tabular}
\end{table}

\begin{table}[ht!]
  \caption{Simulation 3. Results are reported as Mean (SE) where the mean and SE are taken across all frequencies for a given sample size $n$.}
  \label{tbl_sim3}
  \centering
   \begin{tabular}{ccccccc}
        \toprule
        \multicolumn{2}{c}{ } & \multicolumn{5}{c}{SDD} \\
\cmidrule(l{3pt}r{3pt}){3-7}
n & \# True edges & \# Est edges & Precision & Recall & Accuracy & RRMSE \\
\midrule
100 & 27.4 (0.39) & 1.76 (0.12) & 0.82 (0.05) & 0.06 (0.00) & 1.00 (0.00) & 0.99 (0.00)\\
200 & 27.7 (0.20) & 4.80 (0.31) & 0.73 (0.04) & 0.14 (0.01) & 1.00 (0.00) & 0.89 (0.01)\\
500 & 27.9 (0.12) & 6.68 (0.56) & 0.64 (0.04) & 0.16 (0.01) & 1.00 (0.00) & 0.87 (0.01)\\
1000 & 27.9 (0.12) & 8.86 (0.71) & 0.82 (0.03) & 0.29 (0.02) & 1.00 (0.00) & 0.80 (0.02)\\
2000 & 27.9 (0.12) & 22.1 (0.77) & 0.83 (0.01) & 0.65 (0.02) & 1.00 (0.00) & 0.64 (0.01)\\
&   &   &   &   &   &   \\

\multicolumn{2}{c}{ } & \multicolumn{5}{c}{FGL} \\
\cmidrule(l{3pt}r{3pt}){3-7}
 &  & 2950 (157) & 0.01 (0.00) & 0.82 (0.03) & 0.75 (0.01) & 1.00 (0.02)\\
 &  & 2984 (97.0) & 0.01 (0.00) & 0.90 (0.02) & 0.75 (0.01) & 0.90 (0.01)\\
 &  & 7718 (254) & 0.00 (0.00) & 0.97 (0.01) & 0.34 (0.02) & 1.14 (0.03)\\
 &  & 8780 (137) & 0.00 (0.00) & 0.99 (0.00) & 0.25 (0.01) & 0.99 (0.02)\\
 &  & 7975 (156) & 0.00 (0.00) & 0.99 (0.01) & 0.32 (0.01) & 0.72 (0.01)\\
&   &   &   &   &   &  \\

\multicolumn{2}{c}{ }  & \multicolumn{5}{c}{Na\"{i}ve } \\
\cmidrule(l{3pt}r{3pt}){3-7}
 &  & 308 (12.7) & 0.02 (0.00) & 0.26 (0.01) & 0.97 (0.00) & 0.94 (0.01)\\
 &  & 563 (20.2) & 0.02 (0.00) & 0.36 (0.01) & 0.95 (0.00) & 0.91 (0.01)\\
 &  & 2331 (34.6) & 0.01 (0.00) & 0.88 (0.02) & 0.80 (0.00) & 0.81 (0.01)\\
 &  & 3532 (49.8) & 0.01 (0.00) & 0.93 (0.01) & 0.70 (0.00) & 0.71 (0.01)\\
 &  & 5030 (61.8) & 0.01 (0.00) & 0.97 (0.01) & 0.57 (0.01) & 0.61 (0.01)\\
  &   &   &   &   &   &  \\
  
 \multicolumn{2}{c}{ } & \multicolumn{5}{c}{Hard threshold } \\
\cmidrule(l{3pt}r{3pt}){3-7}
 & & - & - & - & - & -\\
 &  & 0.58 (0.05) & 0.29 (0.05) & 0.01 (0.00) & 1.00 (0.00) & 1.74 (0.09)\\
 &  & 0.72 (0.06) & 0.65 (0.05) & 0.03 (0.00) & 1.00 (0.00) & 1.01 (0.00)\\
 &  & 0.81 (0.07) & 0.66 (0.05) & 0.03 (0.00) & 1.00 (0.00) & 0.97 (0.00)\\
 &  & 2.67 (1.78) & 0.65 (0.05) & 0.04 (0.01) & 1.00 (0.00) & 0.96 (0.00)\\
        \bottomrule
        \end{tabular}
\end{table}


\begin{table}[ht!]
\caption{Simulation 1. Difference Between SDD and FGL for Each Metric.}
\label{apdx_tbl_sim1_diff_sdd_naive}
\centering
\begin{tabular}{cccccc}
        \toprule
        n & \# Est edges & Precision & Recall & Accuracy & RRMSE\\
        \midrule
        100 & -2062 (194) & 0.44 (0.06) & -0.58 (0.03) & 0.17 (0.02) & 0.06 (0.01)\\
        200 & -2428 (137) & 0.50 (0.04) & -0.58 (0.02) & 0.21 (0.01) & 0.06 (0.02)\\
        500 & -4952 (326) & 0.64 (0.03) & -0.45 (0.02) & 0.42 (0.03) & 0.10 (0.02)\\
        1000 & -5931 (325) & 0.63 (0.03) & -0.33 (0.02) & 0.51 (0.03) & -0.03 (0.03)\\
        2000 & -6259 (329) & 0.55 (0.03) & -0.27 (0.02) & 0.54 (0.03) & -0.08 (0.03)\\
        \bottomrule
\end{tabular}
\end{table}

\begin{table}[ht!]
\caption{Simulation 1. Difference Between SDD and Na\"{i}ve Method for Each Metric.}
\label{apdx_tbl_sim1_diff_sdd_hard}
\centering
\begin{tabular}{cccccc}
        \toprule
        n & \# Est edges & Precision & Recall & Accuracy & RRMSE\\
        \midrule
        100 & -649 (46.0) & 0.43 (0.06) & -0.40 (0.03) & 0.05 (0.00) & 0.00 (0.00)\\
        200 & -1060 (56.4) & 0.50 (0.04) & -0.50 (0.01) & 0.09 (0.00) & -0.01 (0.01)\\
        500 & -1754 (71.9) & 0.64 (0.03) & -0.40 (0.02) & 0.15 (0.01) & -0.02 (0.01)\\
        1000 & -2258 (103) & 0.62 (0.03) & -0.33 (0.02) & 0.19 (0.01) & -0.07 (0.01)\\
        2000 & -3109 (149) & 0.55 (0.03) & -0.27 (0.02) & 0.27 (0.01) & -0.04 (0.01)\\
        \bottomrule
\end{tabular}
\end{table}

\begin{table}[ht!]
\caption{Simulation 1. Difference Between SDD and Hard Thresholding for Each Metric.}
\label{apdx_tbl_sim1_diff_sdd_soft}
\centering
\begin{tabular}{cccccc}
        \toprule
        n & \# Est edges & Precision & Recall & Accuracy & RRMSE\\
        \midrule
        100 & - & - & - & - & -\\
        200 & 16.0 (1.00) & 0.13 (0.07) & 0.28 (0.02) & -0.00 (0.00) & -0.67 (0.06)\\
        500 & 20.3 (1.15) & -0.03 (0.06) & 0.48 (0.02) & 0.00 (0.00) & -0.22 (0.01)\\
        1000 & 26.1 (1.86) & -0.06 (0.05) & 0.51 (0.04) & -0.00 (0.00) & -0.27 (0.03)\\
        2000 & -182 (52.5) & -0.05 (0.05) & 0.37 (0.05) & 0.02 (0.00) & -0.50 (0.07)\\
        \bottomrule
\end{tabular}
\end{table}

\begin{table}[ht!]
\caption{Simulation 2. Difference Between SDD and FGL for Each Metric.}
\label{apdx_tbl_sim2_diff_sdd_naive}
\centering
\begin{tabular}{cccccc}
        \toprule
        n & \# Est edges & Precision & Recall & Accuracy & RRMSE\\
        \midrule
100 & -4674 (227) & -0.00 (0.00) & -0.50 (0.02) & 0.40 (0.02) & -1.25 (0.16)\\
200 & -5296 (185) & -0.00 (0.00) & -0.60 (0.02) & 0.45 (0.02) & -1.15 (0.12)\\
500 & -10199 (119) & -0.00 (0.00) & -0.97 (0.01) & 0.87 (0.01) & -3.20 (0.07)\\
1000 & -9727 (155) & 0.29 (0.04) & -0.84 (0.02) & 0.83 (0.01) & -2.10 (0.06)\\
2000 & -9189 (177) & 0.33 (0.04) & -0.82 (0.02) & 0.79 (0.02) & -1.35 (0.05)\\
        \bottomrule
\end{tabular}
\end{table}

\begin{table}[ht!]
\caption{Simulation 2. Difference Between SDD and Na\"{i}ve Method for Each Metric.}
\label{apdx_tbl_sim2_diff_sdd_hard}
\centering
\begin{tabular}{cccccc}
        \toprule
        n & \# Est edges & Precision & Recall & Accuracy & RRMSE\\
        \midrule
100 & -257 (10.7) & -0.02 (0.00) & -0.42 (0.01) & 0.02 (0.00) & -0.02 (0.01)\\
200 & -608 (46.2) & -0.02 (0.00) & -0.43 (0.00) & 0.05 (0.00) & -0.01 (0.01)\\
500 & -3564 (56.3) & -0.00 (0.00) & -0.54 (0.01) & 0.30 (0.00) & -0.20 (0.03)\\
1000 & -6483 (58.4) & 0.29 (0.04) & -0.70 (0.02) & 0.55 (0.00) & -0.59 (0.04)\\
2000 & -9139 (98.6) & 0.33 (0.04) & -0.82 (0.02) & 0.78 (0.01) & -1.34 (0.04)\\
        \bottomrule
\end{tabular}
\end{table}

\begin{table}[ht!]
\caption{Simulation 2. Difference Between SDD and Hard Thresholding for Each Metric.}
\label{apdx_tbl_sim2_diff_sdd_soft}
\centering
\begin{tabular}{cccccc}
        \toprule
        n & \# Est edges & Precision & Recall & Accuracy & RRMSE\\
        \midrule
100 & - & - & - & - & -\\
200 & 0.18 (0.14) & -0.06 (0.02) & -0.00 (0.00) & -0.00 (0.00) & -2.05 (0.18)\\
500 & 1.31 (0.36) & -0.46 (0.05) & -0.04 (0.01) & -0.00 (0.00) & -0.09 (0.02)\\
1000 & 2.87 (0.48) & -0.49 (0.06) & 0.02 (0.01) & -0.00 (0.00) & 0.05 (0.01)\\
2000 & 5.44 (0.84) & -0.51 (0.06) & 0.06 (0.01) & -0.00 (0.00) & 0.05 (0.01)\\
        \bottomrule
\end{tabular}
\end{table}


\begin{table}[ht!]
\caption{Simulation 3. Difference Between SDD and FGL for Each Metric.}
\label{apdx_tbl_sim3_diff_sdd_naive}
\centering
\begin{tabular}{cccccc}
        \toprule
        n & \# Est edges & Precision & Recall & Accuracy & RRMSE\\
        \midrule
100 & -2948 (157) & 0.81 (0.05) & -0.76 (0.03) & 0.25 (0.01) & -0.01 (0.01)\\
200 & -2979 (96.9) & 0.72 (0.04) & -0.75 (0.02) & 0.25 (0.01) & -0.01 (0.02)\\
500 & -7711 (254) & 0.63 (0.04) & -0.81 (0.02) & 0.66 (0.02) & -0.27 (0.03)\\
1000 & -8771 (137) & 0.82 (0.03) & -0.71 (0.02) & 0.75 (0.01) & -0.19 (0.02)\\
2000 & -7953 (155) & 0.83 (0.01) & -0.34 (0.02) & 0.68 (0.01) & -0.08 (0.02)\\
        \bottomrule
\end{tabular}
\end{table}

\begin{table}[ht!]
\caption{Simulation 3. Difference Between SDD and Na\"{i}ve Method for Each Metric.}
\label{apdx_tbl_sim3_diff_sdd_hard}
\centering
\begin{tabular}{cccccc}
        \toprule
        n & \# Est edges & Precision & Recall & Accuracy & RRMSE\\
        \midrule
100 & -306 (12.7) & 0.80 (0.05) & -0.20 (0.01) & 0.03 (0.00) & 0.04 (0.00)\\
200 & -558 (20.0) & 0.71 (0.04) & -0.21 (0.01) & 0.05 (0.00) & -0.02 (0.00)\\
500 & -2324 (34.7) & 0.63 (0.04) & -0.72 (0.02) & 0.20 (0.00) & 0.07 (0.01)\\
1000 & -3523 (49.6) & 0.81 (0.03) & -0.65 (0.02) & 0.30 (0.00) & 0.09 (0.01)\\
2000 & -5008 (61.8) & 0.82 (0.01) & -0.33 (0.02) & 0.43 (0.01) & 0.03 (0.01)\\
        \bottomrule
\end{tabular}
\end{table}

\begin{table}[ht!]
\caption{Simulation 3. Difference Between SDD and Hard Thresholding Method for Each Metric.}
\label{apdx_tbl_sim3_diff_sdd_soft}
\centering
\begin{tabular}{cccccc}
        \toprule
        n & \# Est edges & Precision & Recall & Accuracy & RRMSE\\
        \midrule
100 & - & - & - & - & -\\
200 & 4.22 (0.31) & 0.44 (0.06) & 0.13 (0.01) & 0.00 (0.00) & -0.85 (0.09)\\
500 & 5.96 (0.56) & -0.01 (0.06) & 0.14 (0.01) & 0.00 (0.00) & -0.14 (0.01)\\
1000 & 8.05 (0.72) & 0.16 (0.06) & 0.26 (0.02) & 0.00 (0.00) & -0.17 (0.02)\\
2000 & 19.4 (1.87) & 0.18 (0.05) & 0.61 (0.02) & 0.00 (0.00) & -0.32 (0.02)\\
        \bottomrule
\end{tabular}
\end{table}

\end{document}